\documentclass[11pt]{article}
\newif \ifcomments \commentstrue
\newif \ifworkshop \workshoptrue

\workshopfalse

\ifworkshop
\commentsfalse
\fi

\ifworkshop
\newcommand{\silence}[1]{}
\else
\newcommand{\silence}[1]{#1}
\fi

\usepackage{graphicx} 
\usepackage[utf8]{inputenc}
\usepackage{changepage}
\usepackage{url}
\usepackage{hyperref}
\usepackage[letterpaper,margin=1in,bmargin=1.5in]{geometry}
\usepackage{xcolor}
\usepackage{amsfonts}
\usepackage{mathtools}
\usepackage{amsthm}
\usepackage{amssymb}
\usepackage{cleveref}
\usepackage{xspace}
\usepackage{array, caption, tabularx,  ragged2e,  booktabs}
\usepackage{physics}
\usepackage{amsmath}
\usepackage{tikz}
\usepackage{mathdots}
\usepackage{yhmath}
\usepackage{cancel}
\usepackage{color}
\usepackage{siunitx}
\usepackage{array}
\usepackage{multirow}
\usepackage{amssymb}
\usepackage{gensymb}
\usepackage{tabularx}
\usepackage{extarrows}
\usepackage{booktabs}
\usetikzlibrary{fadings}
\usetikzlibrary{patterns}
\usetikzlibrary{shadows.blur}
\usetikzlibrary{shapes}
\usetikzlibrary{positioning,arrows}

\usepackage{algorithm}
\usepackage{algpseudocode}
\usepackage{tablefootnote}

\ifcomments
    \newcommand{\ari}[1]{{\small\textsf{\color{red}{[Ari: {#1}]}}}}
    \newcommand{\andres}[1]{{\small\textsf{\color{blue}{[Andres: {#1}]}}}}
    \newcommand{\mahimna}[1]{{\small\textsf{\color{violet}{[Mahimna: {#1}]}}}}
    \newcommand{\james}[1]{{\small\textsf{\color{olive}{[James: {#1}]}}}}
    \newcommand{\kushal}[1]{{\small\textsf{\color{orange}{[Kushal: {#1}]}}}}
\else
    \newcommand{\ari}[1]{}
    \newcommand{\andres}[1]{}
    \newcommand{\mahimna}[1]{}
    \newcommand{\james}[1]{}
    \newcommand{\kushal}[1]{}
\fi

\newcommand{\util}{{\textsf{util}}}
\newcommand{\true}{{\texttt{true}}}
\newcommand{\false}{{\texttt{false}}}
\newcommand{\tokens}
{{\textsf{tokens}}}
\newcommand{\bribe}
{{\textsf{bribe}}}
\newcommand{\vote}
{{\textsf{vote}}}
\newcommand{\qv}
{{\textsf{quad}}}
\newcommand{\players}{{\cal P}}
\newcommand{\player}{\ensuremath{P}}

\newcommand{\apath}{\mathcal{A}}
\newcommand{\adv}
{{\cal A}}
\newcommand{\indexshort}{\textsf{VBE}\xspace}
\newcommand{\indexlong}{Voting-Bloc Entropy\xspace}
\newcommand{\indexlongbold}{\textbf{V}oting-\textbf{B}loc \textbf{E}ntropy\xspace}
\newcommand{\clustershort}{$\epsilon$-TOC\xspace}
\newcommand{\clusterlong}{$\epsilon$-threshold ordinal clustering\xspace}
\newcommand{\sgn}[1]{\text{sgn}(#1)}

\theoremstyle{definition}
\newtheorem{exmp}{Example}

\newtheorem{definition}{Definition}
\newtheorem{thm}{Theorem}[section]
\newtheorem{lemma}{Lemma}
\newtheorem{corollary}{Corollary}

\newcommand{\sk}{\ensuremath{\mathsf{\sf sk}}}
\newcommand{\pk}{\textsf{pk}}

\newcommand{\gamesfontsize}{\small}
\newcommand{\fpage}[2]{\framebox{\begin{minipage}{#1\textwidth}\gamesfontsize #2 \end{minipage}}}
\title{DAO Decentralization: \\Voting-Bloc Entropy, Bribery, and Dark DAOs }
\author{James Austgen\footnotemark[1] \and Andr\'{e}s F\'{a}brega\footnotemark[1] \and Sarah Allen \and Kushal Babel \and Mahimna Kelkar \and Ari Juels}
\date{Cornell Tech, IC3 \\ 
\smallskip
{\small 1 November 2023 (v1.0)}}

\begin{document}

\maketitle

\def\thefootnote{*}\footnotetext{These authors contributed equally to this work.}\def\thefootnote{\arabic{footnote}}

\begin{abstract}
    Decentralized Autonomous Organizations (DAOs) use smart contracts to foster communities working toward common goals. Existing definitions of decentralization, however—the `D’ in DAO—fall short of capturing key properties characteristic of diverse and equitable participation. 

We propose a new metric called \indexlongbold (\indexshort, pronounced ``vibe'') that formalizes a broad notion of decentralization in voting on DAO proposals. \indexshort measures the similarity of participants’ utility functions across a set of proposals. We use \indexshort to prove a number of results about the decentralizing effects of vote delegation, proposal bundling, bribery, and quadratic voting. Our results lead to practical suggestions for enhancing DAO decentralization.  

One of our results highlights the risk of systemic bribery with increasing DAO decentralization. To show that this threat is realistic, we present the first practical realization of a \textit{Dark DAO}, a proposed mechanism for privacy-preserving corruption of identity systems, including those used in DAO voting. Our Dark-DAO prototype uses trusted execution environments (TEEs) in the Oasis Sapphire blockchain for attacks on Ethereum DAOs. It demonstrates that Dark DAOs constitute a realistic future concern for DAO governance.

\end{abstract}

\section{Introduction}
\label{sec:intro}

A Decentralized Autonomous Organization (DAO) is an entity or community that operates based on rules encoded and executed on a public blockchain~\cite{buterin2013bootstrapping,hassan2021decentralized}. As the name suggests, a DAO's governance is decentralized, meaning that it does not rely on a single individual or highly concentrated authority---in contrast to, e.g., a corporation, where a CEO and board of directors make major decisions. Instead, decisions in a DAO are typically made through community votes on proposals. A DAO's treasury, consisting of crypto assets, also generally resides in its smart contract. The contract enforces adherence to community decisions regarding use of its treasury and also offers operational transparency. 

DAOs can serve many goals, including investment (e.g., The DAO~\cite{Jentzsch2016decentralized,Morris:2023}, Mantle Network~\cite{Mantle:2023}), grant distribution (e.g., MolochDAO~\cite{MolochDAO:2023}, ResearchDAO~\cite{researchdao:2023}), gaming-guild organization (e.g., AvocadoDAO~\cite{Avocado:2023}, GuildFi~\cite{GuildFi:2023}) and---as is the case for DAOs with the largest treasuries---ecosystem governance (e.g., Uniswap~\cite{Uniswap:2023}, Lido~\cite{Lido:2023}, Arbitrum~\cite{ArbitrumDAO:2023}, Optimism Collective~\cite{Optimism:2023}, MakerDAO~\cite{makerdao-whitepaper}). 

DAOs of all types are rising rapidly in popularity. At the time of writing (Nov.~2023), the aggregate value across all DAO treasuries exceeds \$17 billion~\cite{DeepDAO:2023}, almost double the amount just a year ago.

DAOs today vary considerably in their \emph{true} degree of  decentralization. Most have their own associated crypto assets (or ``tokens'') and weigh voting power by token holdings. It is common for vote outcomes to be determined by a small set of ``whales''---a colloquial term used to denote the largest token holders.  Such centralization, as well as low voting participation, are a pervasive source of concern in DAO communities. Vulnerability to centralization has even led to plundering of DAO treasuries~\cite{Malwa:2023}. 

A number of works have sought to recommend ways to improve DAO decentralization.
But first it's necessary to be able to \textit{measure} it in a  way that is reflective of a broad set of real-world concerns. That requirement is the starting point for our work in this paper.

\subsection{Measuring DAO Decentralization}
\label{subsec:measuring_DAO_decentralization}

A common basis for evaluating decentralization in DAOs and other blockchain settings is \textit{token ownership}, specifically the distribution of assets and consequently voting rights among participants~\cite{sharma2023unpacking,fritsch2022analyzing}. Informally, concentration of a large fraction of tokens in a small number of hands---and thus the ability of a small group to determine voting outcomes---is indicative of strong centralization. More widespread distribution, conversely, suggests decentralization.

\textit{Entropy} is one popular metric for measuring decentralization in the distribution of token ownership in a DAO.\footnote{Entropy is typically defined over a random variable. A token ownership distribution may be viewed as a random variable for an experiment where a token is selected uniformly at random and its owner is output.} For a set of addresses $A = \{a_1, \ldots, a_n\}$, where address $a_i$ holds $t_i$ tokens and $T = \sum_{i=1}^ n t_i$:

\[{\sf entropy}(A) \triangleq  -\sum_{i = 1}^n \frac{t_i}{T} \log\Big(\frac{t_i}{T} \Big). \]

\noindent Low entropy corresponds to a high degree of asset concentration and thus strong centralization. High entropy implies the opposite. Other popular decentralization metrics, e.g., the Gini coefficient~\cite{Gini:2023} and the Nakamoto coefficient~\cite{nakamotocoefficient,srinivasan2017quantifying}, are related to various notions of entropy.

Token ownership distribution alone, however, has serious shortcomings as a decentralization metric. To begin with, it is visible on chain only in terms of per-address holdings, not control by real-world individuals. Thus, for instance, an individual who holds 51\% of tokens in a DAO, but spreads them among a large number of addresses could create an appearance of decentralization while having majority control. 

Even if tokens are held by distinct entities, a notion put forward in, e.g.,~\cite{karakostas2022sok}, those entities may have aligned interests and act in concert---a form of centralization. The following examples illustrate cases in which a DAO may be strongly centralized, \textit{even if token ownership appears to imply strong decentralization}.



\begin{exmp}[Low participation / apathy]
\label{exmp:participation}
Lack of participation in DAO governance votes is widespread in practice~\cite{daoapathy} and induces a form of centralization. Consider, for example, a DAO that requires a quorum of 50\% participation for a vote to be ratified. Suppose 50\% of voters do not cast votes and voters other than whales vote 2:1 in favor of the proposal. Whales with just 12.6\% of all tokens can swing the vote and cause the proposal to be rejected.
\end{exmp}

\begin{exmp}[Herding]
\label{exmp:herding}
Interviews with DAO participants have revealed a tendency to vote in alignment with influential community members to preserve reputation~\cite{sharma2023unpacking}, as individual votes are today usually publicly observable. This effect---often called \textit{herding}~\cite{alon2012sequential,banerjee1992simple}---has a centralizing effect. It aligns votes around the choices of a small set of participants. (This problem is similar to the notion of ``herding'' in classical voting theory~\cite{gonzalez2006herding, alon2012sequential}.)
\end{exmp}

\begin{exmp}[Bribery / vote-buying]
\label{exmp:bribery}
Bribery---specifically, \textit{vote-buying}---has been a longstanding concern of DAO organizers~\cite{bribery}. It has a centralizing effect, as it aligns voters around a choice dictated by the briber.
\end{exmp}

Recognizing that token-ownership alone doesn't give a full picture of decentralization, 
researchers have explored broader notions. Most notably, Sharma et al.~\cite{sharma2023unpacking} have considered entropy measures limited to those voters who participate in votes and have also explored graph-based representations of voting patterns (degree centralization, degree assortativity, etc.).  
Token-ownership distribution among voting participants fails to capture important issues, such as those in Examples~\ref{exmp:herding} and~\ref{exmp:bribery}, however, and it's unclear how to interpret graph-based metrics.

With no consensus in the community about how to measure DAO decentralization today, there is a lack of principled guidance on ways to improve DAO decentralization and to combat threats to decentralization, such as vote-buying.

\subsection{\indexlong (\indexshort)}

We introduce a decentralization metric tailored to DAO governance called \textit{\indexlongbold}~(\indexshort, pronounced ``vibe''). \indexshort is based on a foundational principle, that voters with closely aligned interests across elections are a centralizing force, in contrast to the notion of ``credible neutrality'' in voting, which is characterized by ``positive ratings from people across a diverse range of perspectives''~\cite{Buterin:2023}. Expressed differently, the key idea in \indexshort is to \textit{define centralization as the existence of large voting blocs}.

Formally, we express this principle in terms of the \textit{utility functions} of DAO participants, i.e., quantification of the gain or loss associated with voting outcomes. For a given set of elections, a voting bloc is a cluster of voters whose utility functions are similar over outcomes. Utility functions are \textit{latent variables}---conceptually important, but not always directly measurable---and consequently \indexshort is as well~\cite{Latent:2023}.

\indexshort then, measures entropy over voting blocs based on utility functions---rather than over individual token holdings. The result is a broad concept that captures the centralization embodied in all of our examples above. \indexshort is in fact a framework: It allows different notions of clustering and entropy to be plugged in. 

We stress that \indexshort is a theoretical metric: It cannot be measured \textit{directly}, since users do not typically express (or often even know) their utility functions explicitly. But \indexshort provides an important basis for \textit{reasoning about the directional influence of policy choices on decentralization}, and does lend itself to \textit{indirect} measurement.


\paragraph{\indexshort Implications:}
We use \indexshort to prove a number of theoretical results showing how various practices tend to increase or decrease DAO decentralization. (We prove these results relative to particular notions of clustering and entropy.) 

Our main results are as follows:

\begin{itemize}
    \item \textbf{Apathy / inactivity whale:} A large population of apathetic, i.e., non-voting DAO participants, is a centralizing force (Theorem~\ref{thm:apathy}).
    \item \textbf{Delegation:} Given an inactivity whale of large size relative to delegatees, delegation tends (perhaps counterintuitively) to increase decentralization (Theorem~\ref{thm:delegation}).
     \item \textbf{Bribery:} Bribery and decentralization are closely related in the context of DAO governance. The act of bribery decreases decentralization (Theorem~\ref{thm:bribery1}). Additionally, as the decentralization of a DAO rises, so does the risk of systemic bribery---and vice versa (Theorems~\ref{thm:bribery2} and~\ref{thm:bribery3}). 
    
\end{itemize}


We additionally prove results relating to herding and privacy (Theorem~\ref{thm:herding}), \silence{quadratic voting (Theorem~\ref{thm:quadratic-voting})} and owning multiple accounts (Theorem~\ref{thm:owning-multiple-accounts}).

Looking ahead, our theorem statements and proofs are simple, and some just show how \indexshort confirms a known pattern (for example, that quadratic voting is susceptible to sybil attacks). However, our goal is to show the flexibility of \indexshort, and how, unlike prior metrics, it is able to capture the subtle impacts the certain mechanisms have on decentralization. Thus, the main contribution of \indexshort is to put forth a new way to \emph{think} about DAO decentralization. That is, \indexshort introduces a paradigm shift, namely, framing elections in terms of abstract voting entities---instead of individual accounts---as defined by their aligned incentives and interests.

\subsection{Dark DAOs}

Our results on bribery and decentralization (Theorems~\ref{thm:bribery1},~\ref{thm:bribery2}, and~\ref{thm:bribery3}) show that as decentralization increases, bribery can only be successful if it operates on a large scale. 
This observation raises a critical followup question: Is large-scale bribery a realistic future threat to DAOs?

One mechanism postulated for systemic DAO-voting bribery is a \textit{Dark DAO}. A Dark DAO was originally defined as ``a decentralized cartel that buys on-chain votes opaquely.''~\cite{darkdaohack}. Here, ``opaquely'' means that participation in the Dark DAO is confidential. A Dark DAO must also ensure correct execution of a bribery scheme, i.e., bribees are paid if and only if they vote as directed. We define a Dark DAO more broadly as a DAO designed to subvert credentials in an identity system. The goal may be to attack a voting scheme, but could be to attack another system, e.g., we also consider attacks against \textit{privacy pools}~\cite{buterin2023blockchain}, which are effectively DAOs to enhance cryptocurrency privacy.

To date, the feasibility of a fully functional Dark DAO has yet to be demonstrated. In this work, we present a Dark DAO prototype to facilitate vote-buying in DAOs on Ethereum---the most popular blockchain for DAOs today. In its back end, it leverages the confidentiality enforced by trusted hardware (specifically Intel SGX) in the Oasis Sapphire blockchain\footnote{\url{https://github.com/oasisprotocol/sapphire-paratime}}. We also present a ``Dark DAO Lite'' prototype variant that offers greater ease of usability than our basic prototype, but at the cost of weaker confidentiality.

We underscore our belief that Dark DAOs do not pose a current threat, given the limited decentralization of DAOs today, and briefly review ethical considerations in this paper. Our work demonstrates that Dark DAOs are an eminently realistic future threat, however. We discuss possible Dark-DAO mitigations as a first step toward community development of countermeasures. 

Our techniques for Dark-DAO construction are of independent interest, as they point the way toward general techniques for the construction of new financial instruments.

\subsection{Contributions}
In summary, our contributions in this work are:

\begin{itemize}
    \item \textbf{\indexlong~(\indexshort):} We introduce \indexshort, a new metric for DAO decentralization that generalizes prior metrics and addresses a number of their shortcomings (\Cref{{sec:VBE_def}}).
    \item \textbf{Theoretical results:} Using \indexshort, we prove a range of results about how various DAO practices and design choices impact decentralization (\Cref{sec:in-practice}). 
     \item \textbf{Dark DAOs:} Our theoretical results highlight risks of systemic bribery in attacking DAOs---via Dark DAOs---against highly decentralized target DAOs. To show that these risks are a realistic long-term concern, we implement two end-to-end Dark DAO prototypes with different confidentiality / ease-of-use trade-offs (\Cref{sec:dark-daos,sec:Dark_DAO_implementation,{sec:token-based-dd-impl}}). Our techniques are of general interest as they include innovations in the construction of decentralized-finance instruments.
     \item \textbf{Practical guidance:} Based on our theoretical and experimental results, we present concrete points of practical guidance for DAO design and deployment around issues including delegation, voting privacy, voting-slate composition, decentralized identity, and more (\Cref{sec:guidance}). We summarize this guidance  in~\Cref{tab:recommendations}.
\end{itemize}


We \silence{review related work in~\Cref{sec:related} and }conclude with some open research questions in~\Cref{sec:open_questions}.

\section{\indexlong (\indexshort)}
\label{sec:VBE_def}

In this section, we define \textit{\indexlong}~(\indexshort), which sidesteps the aforementioned limitations of prior metrics. It does so by normalizing token holdings based on voters' utility functions.

\paragraph{Intuition.} The key idea behind our definition is to reason about centralization with respect to the tokens held by \emph{groups of DAO members with aligned interests}, instead of with respect to individual members. That is, instead of measuring the distribution of tokens across individual addresses, we focus instead on how tokens are distributed across \emph{blocs} of voters with the same incentives, which are functionally acting as a single entity. Looking ahead, we formalize the notion of ``aligned interests'' by considering the DAO members' \emph{utility functions} across elections.

Aggregating voters based on utility functions allows us to capture the rich interactions and relationships between players in the system, all of which play a role in understanding the true degree of decentralization of a DAO, as discussed in Section~\ref{subsec:measuring_DAO_decentralization}. Indeed, the limitations highlighted there are captured by considering voters with similar utility functions as a single entity; we discuss this extensively in Section~\ref{sec:in-practice}.

\paragraph{DAO abstraction.} We now introduce the notation and formalism that our definition and theorems rely on.

Let $\players = \{\player_1, \ldots, \player_n\}$ be the set of token holders in a system, and $\tokens \colon \players \to \mathbb{R}^+$ a mapping specifying the number of tokens held by each $\player \in \players$. (We will often overload this notation, and input a \emph{set} of accounts to $\tokens$ instead, by which we mean the total tokens held across all accounts in the set). These token holders participate in a set of (binary) elections $E = \{e_1, e_2, \ldots, e_m\}$, where we denote by $\vote_\player \colon E \to \{\true, \false, \bot\}$ player $\player$'s vote in election $e$; $\bot$ indicates that $\player$ abstained from voting in $e$. We define $\util_\player \colon E \times \{\true, \false\} \to \mathbb{R}$ to be the monetary utility of an outcome of $\true$ or $\false$ in $e$ to player $\player$, where we make the simplifying assumption that $\util_\player(e, \true) = -\util_\player(e, \false)$. Player $\player$'s total utility across all elections $E$ is represented by a vector $U_{E, \player} \coloneqq (\util_\player(e_i, \true))_{i \in [m]} \in \mathbb{R}^m$; we denote by $U_{E, \players}$ all players' utilities, i.e., $U_{E, \players} \coloneqq (U_{E, \player})_{\player \in \players}$.

Token holders often have low stakes in the elections, resulting in lack of interest or abstaining from voting altogether. More formally, we say that player $P$ is \emph{$\epsilon$-apathetic} in election $e$ if and only if $|\util_\player(e, \true)| \leq \epsilon$. We denote this set of apathetic voters by $\apath$. If the system supports vote delegation (for example, as a means to combat apathetic voters), players may delegate their tokens to others, who cast a single vote on behalf of all the tokens they now hold.

Lastly, we define $\bribe_\player \colon E \times \{\true, \false\} \to \mathbb{R}$ to be such that it is possible for player $P$ to achieve an outcome of $\true$ (resp., $\false$) in a particular election $e$ via bribery for any expenditure greater than $\bribe_\player(e, \true)$ (resp., $\bribe_\player(e, \false)$). Note that we make the simplifying assumption that bribery costs are independent across elections. We assume that bribing a given $\player$ to flip its vote from $\true$ to $\false$ (respectively, $\false$ to $\true$) costs $\max(2\cdot\util_{\player}(e,\true) + \epsilon, 0)$ (respectively, $\max(2 \cdot \util_{\player}(e,\false) + \epsilon,0)$), for some constant $\epsilon$. Successful bribery to achieve an outcome of $\true$ in $e$ means flipping enough votes to cross a certain threshold $q$ of votes for $\true$ in $e$ (and vice versa for $\false$), i.e., ensuring

\[\sum_{\player \in \players} \Big( \tokens(P) \,\big|\, \vote_P(e) = \true\Big)  > q \cdot \sum\limits_{\player \in \players} \tokens(P).\]

For example, a typical value for $q$ may be $q = 0.5$, which corresponds to an absolute majority. We stress that the equation above represents the threshold to \emph{ensure} the desired outcome in an election, and not just to win it. Albeit related, these notions are not equivalent; the former implies the latter, but not vice-versa. In particular, winning an election is a function of the number of votes cast for the undesired option, whereas ensuring an outcome is agnostic to this.

\subsection{Framework for \indexshort}
We present \indexshort in this section. To do so, we introduce an abstract framework that is parameterized by: (1) a clustering metric, and (2) an entropy notion, which are the two key ingredients that underpin our definition.

\paragraph{Clustering.} We let $C \colon U_{E, \players} \times U_{E, \players} \to \{0, 1\}$ be a clustering function that outputs $1$ if the utilities of two players are ``aligned'' across all elections $E$, and 0 otherwise. Our definition of \indexshort is agnostic to a specific clustering algorithm, and instead only assumes that $C$ specifies an equivalence relation $\sim_{C}$ on the set $\players$, such that $P_i \sim_{C} P_j$ if and only if $C(U_{E, \player_i}, U_{E, \player_j}) = 1$. That is, $C$ partitions $\players$ into classes of players with aligned utility functions across elections. Following standard notation, we denote the set of all classes by $\players/\sim_{C}$, and the class $\player$ belongs to by $[\player]$.

\paragraph{Entropy.} We denote by $F$ a function from the distribution of tokens across \emph{sets} of accounts to real numbers. The purpose of $F$ is to measure, in some sense, how ``evenly distributed'' tokens are across voting blocs. Thus, in practice, $F$ will generally consist of some notion of entropy,\footnote{Entropy is formally defined over a random variable, but we are overloading notation to think of the mapping between sets of accounts and their respective cumulative token balances as the probability mass function of a random variable.} such as one of the many variants of Rényi entropy, e.g., min-entropy, Shannon entropy, or max entropy. (Note that if the blocs are comprised of single entities, this is equivalent to entropy over individual accounts, as defined in prior work~\cite{sharma2023unpacking}.) We stress, however, that in principle $F$ can be any function, and our definition makes no assumptions about its structure. As one example, given a clustering with single entities, we can use any distance metric $d(\cdot, \cdot)$ on $U_{E, \players}$ and define $F$ as the negative of the sum of squares distance between all pairs of player utilities, i.e., $F = -\sum_{\player_1, \player_2 \in \players} d(U_{E,\player_1}, U_{E,\player_2})^2$

\medskip
We are now ready to define \indexshort. Intuitively, our definition says that a DAO is more decentralized if the distribution of tokens across the abstract voting entities specified by $\sim_{C}$ has high entropy according to $F$. More concretely:
\begin{definition}[\indexlong]
\label{def:govdec3}
For a set of elections $E$, a set of players $\players$ with corresponding utilities $U_{E, \players}$, a mapping specifying the distribution of token ownership $\tokens$, a clustering metric $C$, and an entropy function $F$, we define \textit{\indexlong~(\indexshort)} to be: 
\begin{equation*}
\indexshort_{C, F}(E,\players, U_{E, \players}, \tokens) \coloneqq F(\players/\sim_{C}, \tokens).\end{equation*}

\end{definition}

\subsection{Instantiation of VBE}
There are various concrete algorithms with which one can instantiate our \indexshort framework. We propose one such example in this section, which we use throughout the rest of the paper. The advantages and disadvantages of this variant, and \indexshort in general, are discussed in Section~\ref{sec:vbe-practical-considerations}.

\paragraph{Clustering.} We define \emph{\clusterlong~(\clustershort)} as follows:

\begin{equation*}
    C_\epsilon(U_{E, \player_i}, U_{E, \player_j}) \coloneqq \begin{cases} 
      1 & \text{if}\ \forall k \in [m],\
      \Big(\sgn{U_{E, \player_i}[k]} = \sgn{U_{E, \player_j}[k]}\Big) \vee \Big(|U_{E, \player_i}[k]|, |U_{E, \player_j}[k]| \leq \epsilon\Big) \\
      0 & \text{otherwise}
   \end{cases}.\hspace{-2mm}
\end{equation*}

More simply, \clustershort clusters together token holders who have the same preferred outcome across all elections, regardless of how strong this preference is. That is, clusters correspond to token holders whose utility functions are ordinally equivalent. Even though a more granular metric could create clusters based on cardinal utility, we regard ordinal equivalence to be indicative of aligned preferences. Further, as we discuss in Section~\ref{sec:vbe-practical-considerations}, \clustershort has the benefit of being computable based on historical voting data, whereas more complex clustering metrics may be more difficult (or impossible) to estimate.

In addition to these blocs, \clustershort also creates an additional cluster corresponding to all apathetic voters $\apath$, i.e., those whose utilities are close to 0. These voters have aligned preferences, namely, little to no interest in election outcomes.

\paragraph{Entropy.} In this work, we use \emph{min-entropy} as our  entropy notion. That is, for a set of sets of addresses $A$ with a total of $T$ tokens held across all individual accounts,


\[F_{\min}(A, \tokens) \coloneqq \log_2\Bigg(\frac{\max\limits_{A' \in A} \tokens(A'
)}{T} \Bigg). \]

Our entropy notion thus measures the amount of ``information'' in the largest voting bloc by token holdings. As we discuss later on, more granular entropy notions result in more detailed analysis (at the cost of being more difficult to estimate in practice), as these may capture the information in other voting blocs beyond the largest.

\medskip

Putting the two together, we thus get a concrete instantiation of the framework:
\begin{equation*}
\indexshort_{C_\epsilon, \min}(E,\players, U_{E, \players}, \tokens) \coloneqq F_{\min}(\players/\sim_{C_\epsilon}, \tokens) = -\log_2 \Big(\frac{\max\limits_{\players' \in \players/\sim_{C_\epsilon}}\tokens(\players')}{\sum\limits_{\player \in \players} \tokens(P)}\Big).
\end{equation*}

Given that utility functions cannot be measured directly, we cannot compute $\indexshort_{C_\epsilon, \min}$ (or any other variant of \indexshort) directly. However, as we show in the subsequent section, \indexshort can be used as a conceptual tool to reason about the high-level impact that changes in the system (such as the implementation of policy choices) have on the decentralization of a DAO. Namely, one can reason about the \emph{directional} influence of said changes on the utility functions of the players, and thus derive conclusions about whether \indexshort broadly increased or decreased.

Further, utility functions (and, thus, \indexshort) can be estimated via \emph{observable variables}, which can be measured directly. In this case, one can explicitly compute \indexshort, and derive concrete metrics. The degree to which observable variables accurately estimate utility functions (for the purposes of \indexshort) will depend, to a great extent, on the specific clustering algorithm that is used. We discuss this in more detail in Section~\ref{sec:vbe-practical-considerations}.
\section{Implications of \indexshort: Theoretical Results}\label{sec:in-practice}
We now present a variety of theoretical results implied by \indexshort. These results show how, unlike prior notions, \indexshort reflects many of the subtle issues that impact decentralization in a DAO---such as those described in Section~\ref{subsec:measuring_DAO_decentralization}---and thus can serve as a springboard for more accurate understanding of the goals of a DAO.

We first note that, for most ``reasonable'' instantiations of $F$ (such as any Shannon or min-entropy), computing an analogous metric over account balances alone, instead of over voting blocs, gives an upper bound on \indexshort. (Note that the former includes the entropy-based metrics of prior work such as~\cite{sharma2023unpacking}.) Concretely, this fact holds for any $F$ that increases whenever the tokens held by any players's voting bloc increase. More formally:

\begin{lemma}
Let $C_{\texttt{solo}}$ be the clustering metric that partitions $\players$ into singleton sets, i.e., \[C_{\texttt{solo}}(U_{E, \player_i}, U_{E, \player_j}) \coloneqq 1 \iff i = j,\] and let $F$ be any function that is monotonically increasing with respect to $\tokens([\player])$ for every $\player \in \players$. Then, for any clustering metric $C$, it follows that
\[\indexshort_{C_{\texttt{solo}}, F}(E,\players, U_{E, \players}, \tokens) \geq \indexshort_{C, F}(E,\players, U_{E, \players}, \tokens).\]
\end{lemma}
\begin{proof}
    This simply follows from the fact that, for all players $\player$, the number of tokens held by their bloc according to $C$ is necessarily greater than or equal to the number of tokens held by their bloc according to $C_{\texttt{solo}}$: in the former, either $\player$ got grouped in a bloc with more players (and thus holds more total tokens), or she stayed alone in her bloc. As such, by definition, $F$ will increase correspondingly.
\end{proof}

We stress that this lemma holds for most $F$ of practical interest, such as Shannon entropy and min-entropy. As such, \indexshort is, at worst, equivalent to the entropy-based notions introduced by prior work, which focus on account balances. In the subsections that follow, we will show how, in fact, \indexshort reveals more information, as it is able to capture how decentralization is affected by various mechanisms, irrespective of (lack of) fluctuations in account balances. Towards this, we will first present the general recipe of our theorems, before moving on to concrete results.

We note that, for clarity of presentation, our theorems focus exclusively on one instantiation of $\indexshort$, namely, using \clustershort and min-entropy as the clustering metric and entropy function, respectively. However, even though the theorem statements and proof details would differ for other variants of \indexshort, the conceptual takeaways are general (and, in fact, can be made more specific with more granular instantiations of \indexshort).

\subsection{\indexshort Master Theorem}
The theorems in the subsections that follow all aim to show the impact of policy choices or system changes on DAO decentralization, in terms of \indexshort. They all have a similar structure: (1) we consider two systems such that the only difference between them is some ``transformation'' of interest, e.g., a portion of the voters become apathetic, elections are instead private, etc; (2) we reason about the impact of this transformation on the largest voting bloc of both systems; (3) based on this, we compute and compare the \indexshort of both systems.

We now define a ``master'' theorem for $\indexshort_{C_{\epsilon}, \min}$ which captures this template, and thus serves as a proof framework that can be instantiated with concrete transformations of interest. Indeed, our theoretical results that follow are examples of this, as they all invoke this master theorem. (We note that the master theorem can be easily tweaked to accommodate different instantiations of \indexshort, but here we focus exclusively on \clustershort and min-entropy for clarity of presentation.)

In all theorems that follow, we denote by $E$ a set of binary elections, $\players$ a set of players that participate in such elections, $\tokens$ a mapping specifying the number of tokens owned by each player, and $U_{E, \players}$ the players's utilities across the elections. The master theorem then proceeds as follows:

\begin{thm}[\indexlong Master Theorem]
\label{thm:vbe-master}
     We define $T$ to be a function that represents a \emph{system transformation}, i.e., a change in the players, elections, utilities of the players, and/or the distribution of tokens, which we denote by $(\players', E', U'_{E', \players}, \tokens') \coloneqq T(\players, E, U_{E, \players}, \tokens)$. The total number of tokens in the system stays constant, however. Let $B$ and $B'$ be the (not necessarily unique) largest \clustershort clusters by token holdings according to $(E, U_{E, \players}, \tokens)$ and $(E', U_{E', \players}, \tokens')$, respectively. Then, it follows that
    \begin{equation*}
        \tokens'(B') \geq \tokens(B) \iff \indexshort_{C_{\epsilon}, \min}(E, \players, U_{E, \players}, \tokens) \geq \indexshort_{C_{\epsilon}, \min}(E', \players', U'_{E', \players}, \tokens').
    \end{equation*}

    
\end{thm}
\begin{proof}
    This follows directly from the definition of $\indexshort_{C_{\epsilon}, \min}$:
    \begin{alignat*}{3}
        &    &\quad \tokens'(B') \geq&\ \tokens(B) \\
        &  \iff  &\quad \frac{\tokens'(B')}{\sum\limits_{\player \in \players'} \tokens'(\player)} \geq&\ \frac{\tokens(B)}{\sum\limits_{\player \in \players} \tokens(\player)} \\
        &  \iff  &\quad -\log_2 \Big(\frac{\tokens'(B')}{\sum\limits_{\player \in \players'} \tokens'(\player)}\Big) \leq &\ -\log_2 \Big(\frac{\tokens(B)}{\sum\limits_{\player \in \players} \tokens(\player)}\Big) \\
        &  \iff  &\quad \indexshort_{C_{\epsilon}, \min}(E', \players', U'_{E', \players'}, \tokens') \leq &\ \indexshort_{C_{\epsilon}, \min}(E, \players, U_{E, \players}, \tokens)
    \end{alignat*}
\end{proof}
Note that, if $B'$ represents a (new) majority by token holdings, then \indexshort strictly increases; equality follows when $\tokens'(B') = \tokens(B)$.

This master theorem thus serves as a template that individual theorems can bootstrap off of: simply specify a transformation $T$, explain how this modifies the largest voting bloc (if at all), and invoke Theorem~\ref{thm:vbe-master}. Armed with this formula, we now move on to concrete theoretical insights implied by \indexshort. Our theorem statements and proofs are simple, and often just show how \indexshort confirms a known pattern (for example, that quadratic voting is susceptible to sybil attacks). However, our goal is to show the flexibility of (a limited instantiation of) \indexshort, and how, unlike prior metrics, it is able to capture the subtle impacts the certain mechanisms have on decentralization.

\subsection{Owning Multiple Accounts}

As explained in Section~\ref{sec:intro}, previous notions of entropy fail to capture the centralization that is present (but hidden) when a whale distributes tokens across multiple accounts / addresses. In such cases, it may appear that tokens are well diversified across accounts, while a large fraction are in fact under the control of one entity. Unlike prior notions, \indexshort captures this nuance, since these accounts would indeed be considered a single voting bloc (we make the simplifying assumption that an individual's utility function is the same across all her accounts; this need not always be true). We formalize this below.

\begin{thm}[Sybil Attacks and \indexshort]
\label{thm:owning-multiple-accounts}
     Let $(\players', E, U'_{E, \players'}, \tokens') = T_{\texttt{mult}}(\players, E, U_{E, \players}, \tokens)$ be the transformation where some player $\player \in \players$ divides her tokens across a new set of accounts $\hat{\players}$, i.e., $\players' = \players \cup \hat{\players}$, $\tokens'(\hat{\players}) = \tokens(\player)$, and $\forall \hat{\player} \in \hat{\players}$, $U'_{E, \hat{\player}} = U_{E, \player}$. The rest of the system remains unchanged. Then, it follows that
    \begin{equation*}
        \indexshort_{C_{\epsilon}, \min}(E, \players, U_{E, \players}, \tokens) = \indexshort_{C_{\epsilon}, \min}(E, \players', U'_{E, \players'}, \tokens').
    \end{equation*}
\end{thm}

\begin{proof}
Let $B$ be the largest voting bloc by token holdings before $T_{\texttt{mult}}$, which may or may not include $\player$. By assumption, all $\hat{\player} \in \hat{\players}$ are such that $U'_{E, \hat{\player}} = U_{E, \player}$. Thus, all new accounts will be in the same voting bloc $B'$ after $T_{\texttt{mult}}$, namely, $B' = [\player]$.

It follows then that, even though $\player$'s tokens are distributed between all individual accounts in $\hat{\player}$, they are in fact still under the control of the same block, i.e., $B'$. As such, $\tokens'(B') = \tokens(B')$. So, since no blocs acquire any new tokens, $B$ is still the largest voting bloc by token holdings after $T_{\texttt{mult}}$. Then, from Theorem~\ref{thm:vbe-master} it follows that 
\[\indexshort_{C_{\epsilon}, \min}(E, \players, U_{E, \players}, \tokens) = \indexshort_{C_{\epsilon}, \min}(E, \players, U'_{E, \players}, \tokens)\]
as desired.
\end{proof}

This result shows that, according to \indexshort, the ``true'' decentralization of the system does not change when a whale splits her tokens into multiple accounts, as they are all still under the control of the same voting entity. Conversely, metrics that focus on account balances alone would mistakenly conclude that the decentralization of the system strictly increased, since a set of tokens is diversified across more accounts.

\subsection{Apathy} 
\label{subsec:apathy}
A system where voters are apathetic, i.e., not interested in the direction of the community, is not aligned with the goals of a DAO: distribution of tokens is irrelevant if individuals abstain from voting, as elections are narrowed squarely to the set of more invested stakeholders. Our definition captures this fact. Intuitively, apathetic voters all have similar utility functions, which reflect their lack of stake in the elections. \indexshort groups all of these players within the same voting bloc, due to their aligned utilities. (Recall that we use $\apath$ to denote this set of apathetic voters.)

As we formalize below, if the disinterested players are small stakeholders to begin with, apathy has a centralizing effect, as they now belong to a larger bloc of aligned voters. Indeed, in practice, it is common for the set of apathetic voters to represent a majority of token holdings~\cite{daoapathy,fritsch2022analyzing}. We note, however, that interestingly apathy can potentially also have a decentralizing effect, in the (rare) case where it helps diversify a larger coalition of voters.

\begin{thm}[Apathy and \indexshort]
\label{thm:apathy}
    Let $(E, U'_{E, \players}, \tokens) = T_{\texttt{apath}}(\players, E, U_{E, \players}, \tokens)$ be the transformation where players $\hat{\players} \subseteq \players$ become $\epsilon$-apathetic, i.e., $\forall \player \in \hat{\players}$ $\forall e \in E,$ $|\util'_\player(e, \true)| \leq \epsilon$. The rest of the system remains unchanged. Then, if $\forall \player \in \hat{\players},$ $\tokens(\apath) \geq \tokens([\player])$, it follows that
    \begin{equation*}
        \indexshort_{C_{\epsilon}, \min}(E, \players, U_{E, \players}, \tokens) \geq \indexshort_{C_{\epsilon}, \min}(E, \players, U'_{E, \players}, \tokens).
    \end{equation*}
\end{thm}

\begin{proof}
Let $B$ be the largest voting bloc by token holdings before $T_{\texttt{apath}}$. We first note that all apathetic voters belong to the same voting bloc $B'$, according to \clustershort: by the definition of $\epsilon$-apathetic, it follows that, for all $\player_i, \player_j \in \hat{\players}$ and $e \in E$,
\[|\util'_{\player_i}(e, \true)|, |\util'_{\player_j}(e, \true)| \leq \epsilon,\]

which corresponds precisely to the bloc of apathetic voters in \clustershort, containing all players in $\apath$. Then, by assumption, $\tokens(B') = \tokens(\apath) \geq \tokens([\player])$, $\forall \player \in \hat{\players}$. So, since no other blocs decrease in size, it follows that $\tokens(B') \geq \tokens(B)$: either the bloc that aggregates all apathetic voters is now the largest bloc, or the same bloc is the largest in both instances. Thus, from Theorem~\ref{thm:vbe-master}, it follows that
\[\indexshort_{C_{\epsilon}, \min}(E, \players, U_{E, \players}, \tokens) \geq \indexshort_{C_{\epsilon}, \min}(E, \players, U'_{E, \players}, \tokens)\]
as desired.
\end{proof}

This result shows that \indexshort captures the intuition that large-scale apathy, which is common in practice, has a centralizing effect. We refer to the bloc of apathetic voters in a DAO, i.e., non-voting token holders, as the \textit{inactivity whale}. This term reflects the collective and potentially systemically important inactive behavior of this group.

\subsection{Delegation}
\label{subsec:delegation}
Intuition would suggest that delegation leads to a more centralized system: tokens that were originally held by a large set of players, are instead under the control of the (smaller) set of delegates. As we prove formally below, however, \indexshort shows how this situation is more nuanced, as delegation actually tends to make a DAO \emph{more} decentralized: before delegation, the tokens are all held by a \emph{single} voting bloc, namely, the inactivity whale. Delegation then diversifies the tokens held by this ``whale'', and distributes them amongst a set of voting blocs (the delegates). Assuming that the size of the inactivity whale is larger than each delegate's total tokens---which tends to be true in practice~\cite{daoapathy,fritsch2022analyzing}---the system is now more decentralized.

\begin{thm}[Delegation and \indexshort]
\label{thm:delegation}
Let $(E, U'_{E, \players}, \tokens') = T_{\texttt{deleg}}(\players, E, U_{E, \players}, \tokens)$ be the transformation where players $\hat{\players} \subseteq \players$, who are $\epsilon$-apathetic, instead delegate their votes to a set of delegates $D \subset \players$, i.e., $\tokens'(D) = \tokens(\hat{P})$ and $\tokens'(\hat{\players}) = 0$. The rest of the system remains unchanged. Then, if $\forall d \in D,$ $ \tokens(\apath) \geq \tokens'([d])$, it follows that,
    \begin{equation*}
        \indexshort_{C_{\epsilon}, \min}(E, \players, U'_{E, \players}, \tokens') \geq \indexshort_{C_{\epsilon}, \min}(E, \players, U_{E, \players}, \tokens).
    \end{equation*}
\end{thm}

\begin{proof}
    Let $B$ by the largest voting bloc by token holdings before $T_{\texttt{deleg}}$. As discussed in the proof of Theorem~\ref{thm:apathy}, all players in $\hat{P}$ belong to the same voting bloc for all elections in $E$---the inactivity whale---since they are all part of the set of apathetic voters $\apath$. Let $B'$ be the largest voting bloc by token holdings after $T_{\texttt{deleg}}$; note that it may be the case that $B' = [d]$ for some $d \in D$.

    We first note that $B'$ is equal to either (1) $B$ itself, (2) the second largest voting bloc after $B$ before delegation, or (3) $[d]$, for some $d \in D$. That is, since the only blocs that change after $T_{\texttt{deleg}}$ are all the $[d]$ and the inactivity whale (which lost $\tokens(\hat{P})$ tokens), it must the be case that the new largest voting bloc is either the same one as before delegation, the second largest voting bloc before delegation (i.e., $B$ was the inactivity whale, which got fractionated by delegation), or one of the $[d]$ which increased in size.
    
    For (1) and (2), it is clearly the case that $\tokens(B) \geq \tokens'(B')$. Then, for (3), note that, by assumption, $\tokens(\apath) \geq \tokens'([d])$, for all $d \in D$. So, $\tokens(B) \geq \tokens(\apath) \implies \tokens(B) \geq \tokens'([d]) = \tokens'(B')$. 
    
    It follows then that, in all cases, $\tokens(B) \geq \tokens'(B')$. Thus, from Theorem~\ref{thm:vbe-master}, we get that
    \[\indexshort_{C_{\epsilon}, \min}(E, \players, U'_{E, \players}, \tokens') \geq \indexshort_{C_{\epsilon}, \min}(E, \players, U_{E, \players}, \tokens)\]
    as desired.
\end{proof}

The intuition behind this result is that, as long as the delegates are not ``too big'', delegation actually has a decentralizing effect. Conversely, if some delegate is a whale, or gets delegated an overwhelming majority of tokens, then the system may become more centralized. Thus, delegation is most useful in cases where apathy is high. This idea is captured by the following corollary:

\begin{corollary}
If, in Theorem~\ref{thm:delegation}, there exists some delegate $d \in D$ such that $\tokens'([d]) \geq \tokens(\hat{P})$, then $\indexshort_{C_{\epsilon}, \min}(E, \players, U'_{E, \players}, \tokens') \leq \indexshort_{C_{\epsilon}, \min}(E, \players, U_{E, \players}, \tokens)$.
\end{corollary}

In practice, it is common for delegates to be small relative to the inactivity whale~\cite{daoapathy}, but this, of course, need not always be true.

\subsection{Herding}
\label{subsec:herding}
A core goal of DAOs---and any democratic system more broadly---is for token holders to vote according to their true preferences. In practice, however, many DAOs fail to meet this goal and instead exhibit \emph{herding} behavior. Specifically, when votes are publicly observable, social dynamics lead to the formation of ``coalitions'' of voters. For example, token holders have reported feeling influenced to vote a certain way, often in alignment with influential community members, in order to thwart the reputational risks associated with opposing popular viewpoints~\cite{sharma2023unpacking}. Similarly, it has been observed and measured that token holders often vote in alignment with their peers~\cite{messias2023understanding}, who now operate as a single, large entity. In both cases, the monetary utility derived from the social impact of a player's vote skews the monetary utility of her desired outcome in a vacuum.

Herding leads to more centralization, as votes artificially converge on one outcome. Token distribution alone, however, does not show this. Indeed, a system where tokens are distributed evenly, but all players vote for the same outcome due to herding, would be deemed optimally decentralized according to such metrics. Conversely, \indexshort does conclude that reputational risks lead to more centralization, as it aligns the utilities of the players towards the socially preferred outcome.

\begin{thm}[Herding and \indexshort]
\label{thm:herding}
    Let $(E, U'_{E, \players}, \tokens) = T_{\texttt{herd}}(\players, E, U_{E, \players}, \tokens)$ be the transformation where players $\hat{\players} \subseteq \players$ exhibit herding toward, without loss of generality, \true. That is, for all $P \in \hat{P}$ and $e \in E$, the monetary reputational cost of voting for \false\ is greater than or equal to
    $\max(2 \cdot \util_{\player}(e,\false]) + \epsilon, 0)$ for some constant $\epsilon$. The rest of the system remains unchanged. Then, it follows that
    
    \begin{equation*}
        \indexshort_{C_{\epsilon}, \min}(E, \players, U_{E, \players}, \tokens) \geq \indexshort_{C_{\epsilon}, \min}(E, \players, U'_{E, \players}, 
        \tokens).
    \end{equation*}
\end{thm}

\begin{proof}
Let $B$ be the largest voting bloc by token holdings before $T_{\texttt{herd}}$. Note that, after $T_{\texttt{herd}}$, all voters in $\hat{P}$ belong to the same voting bloc $B'$: for every $\player \in \hat{\players}$, $U'_{E, \player}$ will consist of only positive values: either $\player$ preferred an outcome of \true\ in $e$ to begin with, or their monetary utility of \true\ is now $|\util_\player(e, \false)| + \epsilon$. Thus, since $\sgn{\util_\player(e, \true)} = 1$ for all $e \in E$, all of $\hat{P}$ consists of a single voting bloc $B'$ according to \clustershort. 

It follows then that $\tokens(B') \geq \tokens(B)$, as either the ``new'' voting bloc $B'$ is now the largest bloc, or the same bloc is the largest before and after $T_{\texttt{mirr}}$. Then, from Theorem~\ref{thm:vbe-master}, it follows that
\[\indexshort_{C_{\epsilon}, \min}(E, \players, U_{E, \players}, \tokens) \geq \indexshort_{C_{\epsilon}, \min}(E, \players, U'_{E, \players}, \tokens)\]
as desired. 
\end{proof}

An important conclusion of this theorem is that privacy instead \emph{increases} the decentralization of a system, as it serves as a ``mitigation'' to herding. That is, if votes are private, token holders can vote for their true preferences, instead of being influenced by, e.g., social optics or the votes of their peers. (We omit a formal proof of this corollary, as it follows directly via a proof by contradiction of Theorem~\ref{thm:herding}).

\subsection{Voting slates}
\label{subsec:voting-slates}
Grouping together various elections into a lesser number of (more general) elections---so-called ``voting slates''---is in opposition with decentralization: decision-making is more diluted, thus decreasing the relative impact of each voter in the underlying proposals. That is, voting slates ``factor out'' differences in the viewpoints of individuals, yielding more homogeneous utilities. For example, two players may disagree in many of the individual proposals, but agree on a few of the more important ones, resulting in them casting the same overall vote.

We model a player's utility for a slate of elections simply by adding the utilities of the underlying proposals. That is, for all $\player \in \players$ and some election $\mathcal{E}$ comprised of some subset of elections of $E$, the utility of $\player$ in $\mathcal{E}$ is:

\[\util_\player(\mathcal{E}, \true) = \sum\limits_{e \in \mathcal{E}} \util_\player(e, \true).\]

Voting slates are generally used to ``hide'' unpopular or  proposals among a larger set of benign, popular proposals, and thus increase their chances of passing. We model this by saying that if two $\player_i, \player_j$ have aligned utilities (according to \clustershort) on all proposals underlying $\mathcal{E}$, then they will agree on $\mathcal{E}$ itself, i.e., 

\[C_\epsilon(U_{E, \player_i}, U_{E, \player_j}) = 1 \implies \sgn{\sum\limits_{e \in \mathcal{E}} \util_{\player_i}(e, \true)} = \sgn{\sum\limits_{e \in \mathcal{E}} \util_{\player_j}(e, \true)}\]




As we show below, \indexshort reflects the fact that bundling proposals indeed decreased decentralization: by considering a narrower set of elections, which smoothens utility functions, different voting blocs are combined to form larger ones.

\begin{thm}[Voting Slates and \indexshort]
\label{thm:voting-slates}
    Let $(E', U'_{E', \players}, \tokens) = T_{\texttt{slates}}(\players, E, U_{E, \players}, \tokens)$ be the transformation where all elections $E$ are bundled together into slates to form a smaller set of elections $E'$. The rest of the system remains unchanged. Then, it follows that
    \begin{equation}
        \indexshort_{C_{\epsilon}, \min}(E, \players, U_{E, \players}, \tokens) \geq \indexshort_{C_{\epsilon}, \min}(E', \players, U'_{E', \players}, \tokens).
    \end{equation}
\end{thm}

\begin{proof}
Let $B$ be the largest voting bloc by token holdings before $T_{\texttt{slates}}$. Then, note that all players in $B$ are still in the same voting bloc $B'$ after $T_{\texttt{slates}}$: since $C_\epsilon(U_{E, \player_i}, U_{E, \player_j}) = 1$ for every pair of players in $B$, by assumption, it follows that \[\forall \mathcal{E} \in E',\ \sgn{\sum\limits_{e \in \mathcal{E}} \util_{\player_i}(e, \true)} = \sgn{\sum\limits_{e \in \mathcal{E}} \util_{\player_j}(e, \true)}.\]

Conversely, players who did not belong to $B$ may, in fact, join $B'$ after $T_{\texttt{slates}}$: even if the players disagree in some of the underlying proposals for a particular slate $\mathcal{E}$, they may cast the same overall vote for the entire slate. As such, $B'$ contains strictly more players than $B$, which implies that $\tokens(B') \geq \tokens(B)$. Then, from Theorem~\ref{thm:vbe-master}, it follows that
\[\indexshort_{C_{\epsilon}, \min}(E, \players, U_{E, \players}, \tokens) \geq \indexshort_{C_{\epsilon}, \min}(E', \players, U'_{E', \players}, \tokens)\]
as desired.



\end{proof}

\subsection{Bribery}
\label{subsec:vbe-and-bribery}

There is an intuitive relationship between decentralization and bribery, namely, that successful bribery poses a threat to decentralization: in such a case, the entity that acquires the votes of the other players now controls a higher proportion of the total tokens than before. However, traditional decentralization metrics, i.e., based on token distribution across accounts, fail to capture this fact: bribed voters, albeit casting votes as instructed by the briber, still technically hold their tokens. Conversely, \indexshort groups all bribed voters in the briber's bloc, as all bribee's now have aligned utility functions, in line with the bribers desired outcome.

We note that, interestingly, similar to our result from Section~\ref{subsec:apathy}, bribery can have the surprising consequence of leading to a more \emph{decentralized} system, in the case where it fragments a larger bloc (say, the inactivity whale, or some large coalition of voters). However, we ignore this edge case and assume instead that the bloc of bribed voters represents a majority by token holdings. (In particular, for the inactivity whale, it would be rational for all apathetic voters to accept a bribe, in which case the entire inactivity whale is absorbed.) As such, even though bribery need not, unconditionally, increase centralization, it poses a practical \emph{threat} to decentralization.

\begin{thm}[Bribery and \indexshort]
\label{thm:bribery1}
    Let $(E, U'_{E, \players}, \tokens) = T_{\texttt{bribe}}(\players, E, U_{E, \players}, \tokens)$ be the transformation where an entity successfully bribes players $\hat{\players} \subseteq \players$ in elections $E$ to achieve an outcome of, without loss of generality, $\true$. The rest of the system remains unchanged. Then, it follows that
    
    \begin{equation}
        \indexshort_{C_{\epsilon}, \min}(E, \players, U_{E, \players}, \tokens) > \indexshort_{C_{\epsilon}, \min}(E, \players, U'_{E, \players}, \tokens).
    \end{equation}
    
\end{thm}
\begin{proof}
Let $B$ be the largest voting bloc by token holdings before $T_{\texttt{bribe}}$. First, note that, after $T_{\texttt{bribe}}$, all voters in $\hat{\players}$ belong to the same voting bloc $B'$. Recall that, in our DAO abstraction, bribing a player $\player$ to flip its vote in election $e$ from $\false$ to $\true$ costs $\max(2 \cdot \util_{\player}(e,\false) + \epsilon,0)$. So, for every $P \in \hat{\players}$ and $e \in E$, either $\util_P(e, \true)$ was already positive to begin with, or it is now $|\util_{\player}(e,\false)| + \epsilon$. Then, since $\sgn{\util_\player(e, \true)} = 1$ for all $e \in E$, all of $\hat{P}$ consists of a single voting bloc $B'$ according to \clustershort. 

It follows then that $\tokens(B') \geq \tokens(B)$, as either the ``new'' voting bloc $B'$ is now the largest bloc, or the same bloc is the largest before and after $T_{\texttt{bribe}}$. Then, from Theorem~\ref{thm:vbe-master}, it follows that
\[\indexshort_{C_{\epsilon}, \min}(E, \players, U_{E, \players}, \tokens) \geq \indexshort_{C_{\epsilon}, \min}(E, \players, U'_{E, \players}, \tokens)\]
as desired.
\end{proof}

\paragraph{Scale of bribery and decentralization.}
A second, more nuanced observation is that successful bribery must be systemic, i.e., must involve a large number of tokens, if (and only if) a system is highly decentralized. Intuitively, if a DAO is highly centralized, a briber can directly coordinate with a few large players to guarantee an election outcome; or, if the briber is a whale herself, she only needs to bribe a few of the smaller players to accumulate enough tokens to mount a successful attack. Instead, in a more decentralized system, players are smaller, so a briber needs to widen the scale of their attack if they want to win an election. That is, in this case, successful bribery requires large-scale coordination among various smallholders.

\begin{thm}[Internal Bribery and \indexshort]
\label{thm:bribery2}
    Let $(E, U'_{E, \players}, \tokens) = T_{\texttt{bribe}}(\players, E, U_{E, \players}, \tokens)$ be the transformation where $U'_{E, \players}$ is some arbitrary change in the utilities of the players. The rest of the system remains unchanged. Assume that an entity in $\players$ needs to bribe other players holding a total of at least $n_1$ and $n_2$ tokens to guarantee an outcome of $\true$ in elections $E$ before and after $T_{\texttt{bribe}}$, respectively. Then, it follows that

    \begin{equation}
        n_1 > n_2 \iff \indexshort_{C_{\epsilon}, \min}(E, \players, U'_{E, \players}, \tokens) < \indexshort_{C_{\epsilon}, \min}(E, \players, U_{E, \players}, \tokens).
    \end{equation}
    
\end{thm}
\begin{proof} 
We first make the trivial observation that the minimum number of tokens that must be bought to guarantee an election outcome occurs when the bribing entity belongs to the largest voting bloc by token holdings. Let $B_1$ and $B_2$ be such blocs before and after $T_{\texttt{bribe}}$, respectively. By Theorem~\ref{thm:vbe-master}, we get that
    \begin{equation*}
        \tokens(B_2) > \tokens(B_1) \iff \indexshort_{C_{\epsilon}, \min}(E, \players, U'_{E, \players}, \tokens) < \indexshort_{C_{\epsilon}, \min}(E, \players, U_{E, \players}, \tokens).
    \end{equation*}


Then, note that, for $i \in \{1, 2\}$,
\[n_i = q \cdot \sum\limits_{\player \in \players} \tokens(P) - \tokens(B_i)\]

It thus follows that $n_1 > n_2 \iff \tokens(B_2) > \tokens(B_1)$, i.e.,
\[n_1 > n_2 \iff \indexshort_{C_{\epsilon}, \min}(E, \players, U'_{E, \players}, \tokens) < \indexshort_{C_{\epsilon}, \min}(E, \players, U_{E, \players}, \tokens)\]
as desired.


\end{proof}

The theorem above shows how, as a DAO becomes more decentralized, a higher number of tokens need to be corrupted to guarantee an election outcome, since all players are small to begin with. Conversely, in a more centralized DAO, large whales only need to corrupt a few tokens to guarantee their desired election outcome.

This result sheds light on the scale of bribery in the case where the briber is a malicious tokenholder a priori. Conversely, the briber may instead be some external entity. In this case, decentralization also raises the risk of systemic bribery: if there are large players in the system, the briber can directly coordinate with whales to achieve their desired election outcome. If, however, the DAO is highly decentralized, the outcome of the election depends on many stakeholders, which thus requires large-scale coordination among these. More formally:

\begin{thm}[External Bribery and \indexshort]
\label{thm:bribery3}
    Let $(E, U'_{E, \players}, \tokens) = T_{\texttt{bribe}}(\players, E, U_{E, \players}, \tokens)$ be the transformation where $U'_{E, \players}$ is some arbitrary change in the utilities of the players. The rest of the system remains unchanged. Let $n_1$ and $n_2$ be the minimum number of players that an external entity needs to corrupt to guarantee an outcome of $\true$ in elections $E$ before and after $T_{\texttt{bribe}}$, respectively. Then, it follows that

    \begin{equation}
        n_1 > n_2 \iff \indexshort_{C_{\epsilon}, \min}(E, \players, U'_{E, \players}, \tokens) < \indexshort_{C_{\epsilon}, \min}(E, \players, U_{E, \players}, \tokens).
    \end{equation}
    
\end{thm}
\begin{proof}
    This proof is very similar to that of Theorem~\ref{thm:bribery2}. Let $B_1$ and $B_2$ be the largest blocs by token holdings before and after $T_{\texttt{bribe}}$, respectively. By Theorem~\ref{thm:vbe-master}, we get that
    \begin{equation*}
        \tokens(B_2) > \tokens(B_1) \iff \indexshort_{C_{\epsilon}, \min}(E, \players, U'_{E, \players}, \tokens) < \indexshort_{C_{\epsilon}, \min}(E, \players, U_{E, \players}, \tokens).
    \end{equation*}
    Then, note that, for $i \in \{1, 2\}$:
    \[ n_i = \frac{q \cdot \sum\limits_{\player \in \players} \tokens(P)}{\tokens(B_i)}\]

    It thus follows that $n_1 > n_2 \iff \tokens(B_2) > \tokens(B_1)$, i.e.,
    \[n_1 > n_2 \iff \indexshort_{C_{\epsilon}, \min}(E, \players, U'_{E, \players}, \tokens) < \indexshort_{C_{\epsilon}, \min}(E, \players, U_{E, \players}, \tokens)\]
    as desired.
    
\end{proof}
We make the important note that, to acquire a fixed number of target tokens (i.e., in the case where the briber is an external actor), bribing a larger set of smaller players is, in fact, \emph{cheaper} than bribing a smaller set of whales to acquire the same number of tokens. Intuitively, larger players are more ``pivotal''~\cite{weyl2017robustness}, i.e., have a greater influence on election outcomes, and thus are more expensive to bribe. As such, decentralization \emph{decreases} the cost to mount a bribery attack on a DAO. We discuss this in detail in Section~\ref{subsec:Dark_DAO_goals}).

Systemic bribery has long been recognized as one of the main threats to traditional elections, and we have now shown that this is also the case for DAOs. However, bribery is not considered a realistic concern in secret ballot elections, due to the fact that such large scale vote buying would be logistically and economically infeasible to coordinate, and would be traceable. Further, rational vote sellers would simply take the bribe but still vote according to their preferences, instead of following the briber's demands. Looking ahead, we will show in Section~\ref{sec:dark-daos} that, conversely, bribery in DAOs can be done via cost-free, untraceable mechanisms, which guarantee fair exchange. Thus, bribery is a realistic and practical threat for DAO elections.

\subsection{Quadratic Voting}\label{subsec:quad-voting}
Quadratic voting~\cite{lalley2018quadratic} is a voting mechanism that attempts to dilute the influence of whales on election outcomes. To do so, a vote from a player that owns $n$ tokens will only have an impact of $\sqrt{n}$ in the outcome election. At face value, quadratic voting seems to make a system more decentralized: the quadratic ``tax'' is directly proportional to the number of tokens a player owns, which thus shrinks the gap between smaller players and whales. However, quadratic voting is known to be vulnerable to sybil attacks and other forms of malicious coordination~\cite{weyl2017robustness,park2017towards}, and thus may have a \emph{centralizing} effect: players that are in large voting blocs implicitly subvert the quadratic tax due to the fact that their true token count is split amongst all bloc members. As a concrete example, consider a quadratic voting system with no verification of real-world identities. In this case, a whale can divide her tokens amongst multiple accounts, which increases the impact that her votes have on the election outcome. (In fact, as we show in Section~\ref{subsec:Dark_DAO_goals}, a whale can subvert quadratic voting even in the presence of robust identity mechanisms.)

Traditional DAO decentralization metrics fail to capture this attack on quadratic voting, since they do not reason about the relationship between the individual accounts. Conversely, \indexshort would group together all the accounts under the control of the same entity as part of the same voting bloc, and thus concluding that decentralization has decreased. This is analogous to our result from Theorem~\ref{thm:owning-multiple-accounts}, which shows that, in general, splitting tokens across multiple accounts does not increases decentralization. In the case of quadratic voting, this mechanisms actually strictly decreases decentralization, since the votes of this smaller accounts have a higher impact on the election outcome.

\paragraph{Quadratic voting and bribery.} 
Similar to a whale splitting off her tokens into multiple accounts, a set of colluding players can have a greater impact in the election outcome if quadratic voting is employed. Namely, quadratic voting may have the surprising consequence of \emph{decreasing} the cost of bribery. The high-level idea is that quadratic voting ``amplifies'' the power of small accounts, which may be cheaper to bribe. Thus, for the same cost, a briber is able to have a greater impact on the outcome of an election.

Prior work has informally identified these issues in the context of traditional elections~\cite{weyl2017robustness,park2017towards}, and DAOs specifically~\cite{philquadvotingdaos}. For the former, as discussed at the end of Section~\ref{subsec:vbe-and-bribery}, large-scale collusion is not considered a realistic threat, and thus the fact that bribery can have a bigger impact on an election outcome if quadratic voting is employed is not seen as a practical limitation. However, since our Dark DAO prototype from Section~\ref{sec:dark-daos} makes bribery inexpensive and guarantees fair exchange, bribery poses a realistic threat to quadratic voting (and any blockchain-based voting scheme for that matter).

Our formalism captures this relationship between quadratic voting and bribery. We define ``small'' accounts to be, concretely, those whose fraction of the total tokens increases with quadratic voting in place, and thus have their impact amplified. More formally, we denote that a player $\player \in \players$ benefits from quadratic voting by $\qv(\player, \tokens) = 1$, where 

\[\qv(\player, \tokens) = 1 \iff \frac{\tokens(\player)}{\sum\limits_{p \in \players} \tokens(p)} < \frac{\sqrt{\tokens(\player)}}{\sum\limits_{p \in \players} \sqrt{\tokens(p)}}.\]

The relationship between quadratic voting and bribery hinges on whether the cost to bribe a player is the same with or without quadratic voting. That is, does the monetary utility of a player in an election change if the impact of their vote changes? If not, then the cost to bribe a small player is the same, but the impact is greater. Otherwise, the cost to bribe the player increases as their relative power increases, which offset each other.

Whether quadratic voting changes a player's utility or not will vary across systems. Broadly speaking, if DAO members take governance seriously and are invested in election outcomes, quadratic voting indeed changes utilities: since smaller accounts become more ``pivotal'' as a result of quadratic voting, their utilities increase correspondingly. Conversely, if members have little interest in governance, the fact that their vote can now have a greater impact in the election will not change their utilities. As such, the nature of a community must be taken into account when deciding to use quadratic voting.

\begin{thm}[Quadratic Voting and Bribery]
\label{thm:quadratic-voting}
Let $(E', U_{E', \players}, \tokens) = T_{\texttt{quad}}(\players, E, U_{E, \players}, \tokens)$ be the transformation where all elections $E$ employ quadratic voting. We denote the election corresponding to $e \in E$ by $e' \in E'$. Let $f$ and $f'$ be the fraction of total votes that a bribing entity is able to control for some fixed expenditure $t$ in elections $E$ and $E'$, respectively. Then, it follows that
    \begin{equation*} 
        f < f' \iff \exists \hat{\players} \subseteq \players \mid \forall \player \in \hat{\players},\ \Big(\qv(\player, \tokens) = 1 \wedge U_{E, \player} = U_{E', \player}\Big)
    \end{equation*}
\end{thm}

\begin{proof}
Recall from Section~\ref{sec:VBE_def} that the cost of bribing all players in $\hat{P}$ to vote for, without loss of generality, $\false$ in all elections $E$ is
\[t = \sum\limits_{p \in \players} \sum\limits_{e \in E'} \max(2\cdot\util_{\player}(e,\false) + \epsilon, 0)\]

Since, by assumption, $\forall \player \in \hat{\players}\ \forall e \in E,\ \util_{\player}(e,\false) = \util_{\player}(e',\false)$, it follows that the cost to bribe all players in $\hat{P}$ across all elections $E'$ is also $t$. Then, since all players in $\hat{P}$ benefit from quadratic voting, the bribing entity has thus managed to control a larger fraction of the total votes for this same expenditure: by definition, for all $\player \in \players$,

\begin{alignat*}{3}
&    &\quad \frac{\tokens(\player)}{\sum\limits_{p \in \players} \tokens(p)} <&\ \frac{\sqrt{\tokens(\player)}}{\sum\limits_{p \in \players} \sqrt{\tokens(p)}} \\
&  \iff  &\quad \sum\limits_{\player \in \hat{\players}} \Big(\frac{\tokens(\player)}{\sum\limits_{p \in \players} \tokens(p)} \Big)<&\ \sum\limits_{\player \in \hat{\players}} \Big(\frac{\sqrt{\tokens(\player)}}{\sum\limits_{p \in \players} \sqrt{\tokens(p)}}\Big) \\
&  \iff  &\quad f <&\ f'
\end{alignat*}
as desired.
\end{proof}

This result thus shows that quadratic voting may be favorable for a bribing entity. In particular, if there are enough small voters whose utilities are unchanged, the cost to guarantee successful bribery decreases:
\begin{corollary}
Assume that, for $\hat{P}$ as defined in Theorem~\ref{thm:quadratic-voting}, $\tokens(\hat{P}) > q \cdot \sum_{\player \in \players} \tokens(P)$. Let $t$ and $t''$ be the expenditure required to guarantee an outcome of $\true$ in elections $E$ and $E'$, respectively. Then, it follows that $t' < t$.
\end{corollary}
This corollary simply follows from the fact that, as proved in Theorem~\ref{thm:quadratic-voting}, the expenditure $t'$ required to control a fraction of $q$ votes in $E'$, and thus guarantee successful bribery in $E'$, would only be enough to acquire a fraction of $q - \epsilon$ votes in $E$. As such, some additional expenditure is required to cross the threshold of $q$ votes.
\section{Practical Considerations of \indexshort}
\label{sec:vbe-practical-considerations}
Since \indexshort is a theoretical metric, it serves primarily as a conceptual tool to reason about how certain decisions or events may impact the decentralization of DAOs. However, \indexshort can also be estimated by using real-world data alongside a particular instantiation of our framework. In this section, we discuss some directions towards this, and the limitations of this approach (and our model more broadly). We note, however, that an empirical study of DAOs, such as concretely computing \indexshort for popular DAOs based on on-chain data, is left as future work.

\paragraph{Limitations of our formal model.} Our DAO abstraction from Section~\ref{sec:VBE_def}, which underpins \indexshort, makes several assumptions, which need not be true in practice. For example, we only consider binary elections, whereas DAO proposals can involve multiple options, e.g., Optimism's process for (retroactively) funding public goods, where votes explicitly expressed the allocation of funds to different organizations~\cite{OptimismPGF}. We note, however, that our model can be reframed to consider arbitrary elections, as this would simply involve using a clustering metric for \indexshort that takes into account multiple potential election outcomes when partitioning $\players$.

A more important limitation is the fact that we assume token holdings remain constant across elections. For simplicity and clarity of presentation, we deem this to be sufficient due to the conceptual nature of our theoretical results. Further, our model can be modified to assume variable token holdings. For example, $\tokens(\player)$ can be defined to be the maximum number of tokens held by $\player$ at any point across all elections $E$. We leave such extensions as future work.

\paragraph{Measuring \indexshort.} As we have emphasized throughout this work, \indexshort is a theoretical notion and cannot be measured directly. This is due to the fact that utility functions are \emph{latent variables}~\cite{Latent:2023}, which are not directly measurable. This is an inherent limitation of any metric that depends on utility functions, including important results and models from voting theory, e.g.,~\cite{weyl2017robustness,duffy2008beliefs}.

Due to this limitation, \indexshort is most useful as a theoretical tool to estimate the directional impact of policy choices in decentralization. However, \indexshort does lend itself to indirect measurement: latent variables, such as utility functions, can be estimated via \emph{observable variables}, which are indeed measurable. In our context, observable variables may be gathered from on-chain data (such as past voting history), low-cost straw polls, social dynamics, etc.

The accuracy of estimating \indexshort via observable variables will depend, to a large extent, on the specific clustering algorithm with which \indexshort is instantiated, and how much ``information'' is required from the utility funtions in order to partition $\players$. For example, a trivial partition of $C(U_{E, \player_i}, U_{E, \player_j}) = 1 \iff U_{E, \player_i} = U_{E, \player_j}$ will yield a less accurate estimate than some other function which only takes into account the voting history of the players.

Practitioners can thus use \indexshort to derive concrete metrics, and analyze the real-world behavior of a system, by initialize our framework with a particular clustering metric and entropy function, and using observable variables to estimate utility functions of players. Deciding which clustering and entropy functions to use requires careful consideration of what data is available. In general, more granular notions of \indexshort are, of course, more informative. For example, Shannon entropy yields more nuanced results than min-entropy, and a clustering metric based on cardinal utilities will result in blocs of players that are more closely aligned. However, such functions may require data that is not easy to gather, or may not even be tractable at all. As such, there is a trade-off between how informative \indexshort is, and how easy it is to compute.

In the particular case of $\indexshort_{C_e, \min}$, for example, historical voting data is sufficient for \clustershort, since clusters are assigned based on ordinal utility. If we assume voters are rational, we can extrapolate this from the casted votes: for any player $\player$ and election $e$, it follows that

\begin{align*}
\vote_\player(e) = \true &\iff \util_\player(e, \true) > \epsilon \\
\vote_\player(e) = \false &\iff \util_\player(e, \false) > \epsilon \\
\vote_\player(e) = \bot &\iff |\util_\player(e, \true)| < \epsilon.
\end{align*}

Even though we cannot extrapolate the exact value of $\util_\player(e, \true)$ based on election outcomes, the equations above are sufficient to use $C_\epsilon$ to cluster players. We stress, however, that different instantiations of the \indexshort framework may require different measurement techniques.

Another natural limitation of \indexshort is that, given that it is a framework, two instances of \indexshort are not directly comparable. That is, in order to reason about the relative level of decentralization between two DAOs, or how decentralization has fluctuated over time for a single DAO, the same variant of \indexshort must be used. In practice, however, we expect that broad \indexshort adoption would involve a handful of standard parameters agreed upon by the DAO community.

\paragraph{Limitations of $\indexshort_{C_e, \min}$.} In addition to the general limitations described above, each variant of \indexshort may pose additional constraints. In the case of $\indexshort_{C_e, \min}$, we lose most of the information provided by all voting blocs except the largest one, since min-entropy is only a function of the latter. This does not imply that the analysis is not accurate, as all subsequent blocs are strictly smaller than the one our definition focuses on,  but rather that other entropy notions may yield additional insights; indeed, min-entropy is always less than or equal to Shannon entropy and max-entropy~\cite{enwiki:1163772345}.

Our clustering metric, \clustershort, is also quite strict, as it does not reason about voters who are aligned in most (but not all) elections. One could instead consider, for example, a generalization \clusterlong that is parametrized by the fraction of elections two players must agree on to be considered part of the same cluster. We opted for the simpler variant in this work, as it serves as a proof-of-concept for our theoretical results; more general clustering metrics would yield the same conceptual conclusions, while making the theorem statements and proofs more opaque with orthogonal mathematical details.

\paragraph{Data Collection.} Even though blockchain-based elections are public, extrapolating relevant observable data for analyzing \indexshort (and other decentralization metrics) is surprisingly difficult. Indeed, as prior work has also pointed out, ``in practice it is not trivial to acquire all governance related information from raw blockchain data''~\cite{feichtinger2023hidden}. To aid the analysis and computation of \indexshort, we thus propose that DAOs publish relevant statistics in a way that is easy to understand and use.

We propose that DAOs choose and specify a variant of \indexshort to support, which then guides how to present voter data. For \clustershort, DAOs can keep a log mapping all token holders to a list of the elections they were eligible for, denoting their vote (if any). Other variants of \indexshort may require more detailed information. Feichtinger et al.~\cite{feichtinger2023hidden} successfully extrapolated a vast amount of governance-related data from 21 DAOs along multiple axes (albeit noting that it was surprisingly challenging), and open-sourced their data set in the form of ``subgraphs'' from The Graph protocol~\cite{thegraph}. Their work may serve as inspiration for user-friendly ways to present voter data: each DAO could implement a publicly-accessible subgraph of governance data, maintained either by a dedicated set of curators or by the community more broadly.
\section{Dark DAOs: Overview}\label{sec:dark-daos}

A \textit{Dark DAO} is itself a DAO, but one whose objective is to subvert a system of decentralized credentials and thereby to target, e.g., voting in another DAO or DAOs. Dark DAOs were first introduced in a 2018 blog post~\cite{darkdaohack}. 
We have shown in~\Cref{sec:in-practice} that as decentralization increases, the cost of a bribery attack rises. Consequently, the need arises for a briber to perform broad coordination, as there is a need to target more users. Thus the threat of Dark DAO deployment increases.

In this section, we briefly explain what Dark DAOs are, giving an informal definition in~\Cref{subsec:Dark_DAO_def}. We outline their main design principle, \textit{key encumbrance}, in~\Cref{subsec:key_encumbrance}. We explain the various ways in which they can disrupt votes in target DAO~\Cref{subsec:Dark_DAO_goals}. 

\subsection{Dark DAO Definition}
\label{subsec:Dark_DAO_def}




A \textit{Dark DAO} is defined specifically in~\cite{darkdaohack} as a ``decentralized cartel that buys on-chain votes opaquely.'' We believe that a broader definition is more informative---one that encompasses any corruption of any system of credentials, whether used for voting or other purposes.
Like an ordinary DAO, a Dark DAO designed to be trust-minimized: it ensures that a bribe is ``fair,'' i.e., a bribee receives money from a briber iff the briber gains agreed-upon access to the bribee's credential(s). Additionally, a Dark DAO is ``opaque'' in the sense of ensuring that participation is \textit{private}.

Informally, then, a Dark DAO has the following three key properties: 
\begin{enumerate}
    \item \textbf{Opacity:} Participants in a Dark DAO are indistinguishable on chain from other credential holders. (Consequently, statistics like the number of participants in a Dark DAO are also hidden.)
    \item \textbf{Fair exchange:} Once a bribee commits to accepting a briber's offered bribe, the briber obtains access to the bribee's credential and the bribee is paid the bribe. 
    \item \textbf{Bounded scope:} A bribee who participates in a Dark DAO contributes no resource to the Dark DAO beyond a committed credential and pre-agreed-upon costs. (E.g., the bribee may also pay normal transaction fees.) 
\end{enumerate}

\paragraph{Example (voting):} A Dark DAO that aims to corrupt voting in a target DAO would involve voters (bribees) selling their votes to a vote buyer (briber). \textit{Opacity} would mean that bribees are indistinguishable from other voters in the target DAO. \textit{Fair exchange} would mean that the briber pays a pre-agreed-upon amount to a bribee iff the bribee's vote is cast as the briber prescribes in a particular election. Finally, \textit{bounded scope} means in this context that the bribee can use her voting credential in an unrestricted way outside of the election in question. 

\paragraph{Remark:} Fair exchange requires not just that a briber gain access to a user's credential, but that the credential be usable in a pre-agreed upon way. For example, if the briber gains access to the bribee's credential, and the bribee is paid, but bribee can revoke the credential before use by the briber, then the exchange is not fair. To capture such subtleties, a more formal definition of fair exchange may be couched in terms of a universe of possible target-system states ${\cal S}$ and a transcript $T = \{S_1, S_2, \ldots\}$ for $S_i \in {\cal S}$ of state transitions. Fair exchange means that for ${\cal S}$ and a set of transcripts ${\cal T}$---both agreed upon by the briber and bribee---$T \in {\cal T}$ for the  transcript $T$ of the target system's history. We defer the development of such formalism to future work. 

\subsection{Main tool: Key encumbrance}
\label{subsec:key_encumbrance}

The main mechanism by which a Dark DAO achieves its properties is \textit{key encumbrance}~\cite{kelkar2023complete}. Key encumbrance is a form of life-cycle management for keys that facilitates selective delegation. In the case of voting, it enables a Dark DAO to ensure that delegated keys are used to cast the votes to which bribees have committed, but gives the Dark DAO no further control over encumbered keys. 

There are two main tools that can enforce such delegation in principle in a way that does not require use of a trusted third party: secure multiparty computation (MPC)~\cite{goldreich-book} and trusted execution environments (TEEs)~\cite{mckeen2013innovative,mckeen2016intel}. 

TEEs are the more practical, particularly as we demonstrate below that existing hardware-based TEE systems are sufficient to realize Dark DAOs. Such TEEs enable applications to run in an integrity- and confidentiality-protected environment.

Dark DAO running in a TEE creates an encumbered key, or imports an already encumbered key $\sk$. 
The briber and the bribee can request use of $\sk$ from the dark dao, which ensures 
compliance with the Dark DAO policy, i.e., enforces the properties enumerated in~\Cref{subsec:Dark_DAO_def}. In the particular case of voting, the briber can use $\sk$ as a voting credential, while the bribee can use $\sk$ to manage their cryptocurrency and interact with smart contracts in a way that does not violate the fair exchange property.

\subsection{Dark DAO goals}
\label{subsec:Dark_DAO_goals}
Globally, the goal of a Dark DAO is to subvert voting in a target DAO. There are a number of ways in which it can do this, of which we enumerate several here. 

\paragraph{Vote buying:} A briber desiring a given election outcome (e.g., a ``yes'' vote) can simply offer payment for votes toward this outcome---where payments are scaled to the weight associated with a given bribee's vote (e.g., proportional to her DAO token holdings). Various forms of conditional payment are also possible, e.g., paying bribes only if the desired outcome is achieved or offering a fixed payment for distribution across the total population of bribees. 

We note too that vote buying works not just for systems in which votes are weighted by token holdings, but also ``one-vote-per-person'' systems, e.g.,~\cite{Napolitano:2023,Worldcoin:2023}. In such cases, an encumbered key $\sk$ might be a user's credential in a decentralized identity system, e.g., in Gitcoin Passport or Worldcoin.

A Dark DAO can further increase the threat of so-called cost-less bribery. For instance, B{\'o}~\cite{bo2007bribing} introduces ``pivotal'' bribes as a way to bribe voters at virtually no cost. Consider a binary vote where the final result is the option chosen by the majority of voters (for simplicity, assume an odd number of users $n$) and suppose the utility of a user for a ``yes'' vote is $U$ (and for a ``no'' vote is $-U$). A briber, wishing to flip votes to ``no'' bribes the user as follows: If the user votes ``no'' and there are exactly $(n+1)/2$ ``no'' votes (i.e., the voter is ``pivotal'' in the sense that the outcome would have changed if the voter chose ``yes'' instead), then the briber pays $P + \epsilon$ (for any $\epsilon > 0$). Otherwise, if the user votes ``no'' (regardless of the result), a bribe of $\epsilon$ is still paid. No bribe is paid if the user votes ``yes.'' It is easy to see now that it is always individually rational for a user to take such a bribe---regardless of the result, the user's utility will be $\epsilon$ larger than if the bribe was not taken. But this means that if all users are rational, and the bribe is offered to everyone, then all users will vote ``no'' resulting in no pivotal voters and the cost of bribing being only $n \epsilon$ (which can be arbitrarily small). A major practical hurdle in deploying such a bribe is conducting is coordinating enough voters, which is made significantly easier by a Dark DAO. In essence, a Dark DAO can make such a ``pivotal''-bribe attack extremely cheap. 

\paragraph{Coordinated price manipulation:} As noted in~\cite{darkdaohack}, it is possible for a Dark DAO to operate without an explicit party distributing bribes. A Dark DAO can instead orchestrate collective action that rewards participants indirectly. 

For example, a Dark DAO can orchestrate the following steps among a cabal: (1) Purchase a collective short position in a target asset $X$; (2) Vote collectively toward an outcome that causes the price of $X$ to drop; (3) Close the short position at a profit; and (4) Distribute profits among Dark DAO participants. Dark DAO goals can in principle extend beyond voting to other actions as well, such as coordinated attacks on consensus protocols~\cite{darkdaohack} or---if the Dark DAO ingests assets from participants---market manipulation.

\paragraph{Undermining perceived election integrity:} The mere presence of a Dark DAO may be enough to cast doubt on a DAO election and call into question whether it is being attacked. DAO opacity conceals the size of a Dark DAO such that even with limited participation, a Dark DAO could impact community trust in an election. 

Alternatively, a Dark DAO could selectively reveal (and prove) statistics---e.g., participation of at least 10\% of token holdings---that would substantiate the threat it poses.

\paragraph{Exploiting quadratic voting and quadratic funding:} 
Quadratic voting~\cite{lalley2018quadratic} is a mechanism that seeks to limit the influence of whales in determining election outcomes. It weights a given voter's vote as the square root of her token balance.

Quadratic voting is only enforceable if tokens are assignable to real-world identities. For instance, if votes are weighted as the square-root of token holdings by address, a whale can boost her voting weight by dividing her tokens among multiple accounts.

A Dark DAO can subvert quadratic voting \textit{even when vote is conducted using a secure decentralized identity system}. That is because a Dark DAO can encumber keys not just in a way that enforces voting choices, but \textit{also} use of digital assets. 

A whale can thus subvert a quadratic voting scheme as follows. The whale does not just bribe voters to vote for a particular outcome, but also \textit{temporarily deposits some of the whale's funds with them}. As bribees' keys are encumbered, the Dark DAO can ensure not just that they vote as directed, but that they will return the whale's funds. 

For example, a whale with 256 tokens can deposit 4 tokens with each of 63 distinct bribees. The result would be an increase in the whale's voting weight of a factor of $(\sqrt{4} \times 64) / \sqrt{256}  = 8$.

A similar attack is possible against quadratic funding~\cite{Buterin-qf}.

\paragraph{Subverting privacy pools:} \textit{Privacy pools}~\cite{buterin2023blockchain} aim to strike a balance between privacy and accountability in privacy coins and privacy services for cryptocurrencies. A privacy pool is an association of users of a certain class that acts as an anonymity set for members' transactions.

For example, a pool might require members to prove that they are not on a sanctions list (e.g., from the U.S. Office for Foreign Asset Control (OFAC), a requirement for most banks~\cite{frey2019sanctions}). Pool membership then implies sanctions compliance, enabling a user to provide assurance that she is not sanctioned, while still preserving her privacy. 

Any set of users may choose to create a pool. Membership requirements for a pool are determined by the community making up the pool. In this sense, a pool is like a DAO. It is also subject to attack by a Dark DAO.

A Dark DAO can target a privacy pool by facilitating {\em identity-selling} and thus selling of access to a pool. A privacy-pool member (``lender'') can sell temporary access to her pool-compliant address to an adversary (``borrower'') who isn't eligible for pool membership. For example, a user in a sanctions-compliant pool can sell pool access to another user who is in fact on a sanctions list. To do so, the seller encumbers her address so that it is subject to limited control by the buyer.

\bigskip
\noindent \textit{Example:}
Alice holds a Dark DAO address $a$ that is a member of sanctions-compliant privacy pool $P$. Mallory will be receiving money from an address $z$. Alice agrees to help Mallory launder the money through $P$. 

Alice sets the Dark DAO policy for address $a$ so that when funds are received from $z$: (1) 99\% of funds are subject to control by Mallory, i.e., Mallory can send those funds from $a$ to any other desired address and (2) 1\% of funds are subject to control by Alice---as payment for Mallory's borrowing of $a$.
\bigskip

Interestingly, a Dark DAO can conversely \textit{reinforce} the security of a privacy pool by enforcing a policy in which members' addresses must maintain a minimum balance for some period of time. Such a policy would limit weakening or dissolution of the pool.

\section{Basic Dark DAO}
\label{sec:Dark_DAO_implementation}

To illustrate that Dark DAOs are practical, we have implemented a Dark DAO prototype. Our implementation is  written in what is currently the most popular smart-contract programming language, Solidity. Furthermore, while it uses TEEs, it demonstrates that developing a Dark DAO need not require any special knowledge of TEEs, thanks to the abstractions provided by the TEE-based Oasis Sapphire blockchain.

Our prototype demonstrates in particular how Dark DAOs might coordinate bribery on a popular off-chain voting platform, Snapshot~\cite{snapshotwebsite}. To the best of our knowledge, however, all current DAO voting platforms are susceptible to Dark DAO interference.
Our open-source implementation can be found at \url{https://github.com/DAO-Decentralization/dark-dao}. 

In what follows, we first give a background on Oasis Sapphire and Snapshot (Section~\ref{sec:impl-bg}) followed by the details of our Dark DAO design (Section~\ref{sec:impl-design}) and possible future design enhancements. We then discuss the cost of participation (Section~\ref{sec:dd-cost}), security (Section~\ref{subsec:basic_DAO_conf}), deployment considerations (Section~\ref{sec:impl-depl-consider}), mitigations against negative impacts of Dark DAOs (Section~\ref{sec:dd-mitigation}), and ethical the considerations of building a Dark DAO prototype (Section~\ref{sec:dd-ethical}).

\subsection{Background}
\label{sec:impl-bg}
\paragraph{Oasis Sapphire.} Oasis Sapphire is an implementation of the Ethereum Virtual Machine that runs entirely inside a TEE. Assuming the TEE is not broken, Sapphire is able to execute smart contracts whose state is kept private both during execution (in memory) and after execution (in storage). In addition to matching the base implementation of the EVM, Sapphire also includes several built-in precompiled contracts that make it both easy and cheap to perform cryptographic operations pertinent to Dark DAOs: generating entropy and signing messages. These methods are not available in blockchains that do not use TEEs or encrypted computation, since the private key material or entropy would necessarily be leaked to all blockchain nodes.

In its current state, Sapphire does not provide confidentiality for all transaction metadata: senders and recipients of every transaction are public. Additionally, its persistent storage lies outside the TEE, making contract storage access patterns a vector for information leakage~\cite{jeanlouis2023sgxonerated}. In this paper, we do not address side channel attacks such as these.


Sapphire is compatible with many cryptocurrency wallets today, making it a good candidate for hosting a key encumbrance system.

\paragraph{Snapshot.} Snapshot is an open source, centralized, off-chain voting platform for DAOs. Rather than requiring DAO users to pay the costs of making on-chain voting transactions, it accepts votes submitted as signed messages to the Snapshot website. The website is organized into ``spaces,'' typically one per DAO, each of which is moderated and/or controlled by a permissioned hierarchy of administrators and moderators of the DAO. At the time of writing, there exist over 28,000 Snapshot spaces~\cite{snapshotwebsite}. 

Individual DAOs are free to adjust the algorithm used to calculate the weight of an individual vote, termed its \textit{voting power}. How much voting power a particular user gets often is determined by how many DAO tokens he or she is holding on a blockchain at the moment a proposal is created, and thus a ``snapshot'' of voting power is taken at the corresponding block. For a given DAO proposal, the signed voting messages are collected and, once the voting period is over, are published as receipts to IPFS~\cite{snapshottechnicaloverview}, a distributed file sharing network. Voters can verify that their votes were included in the proposal's outcome by checking their voting receipts.

Snapshot also provides a means for \textit{delegating} one's voting power to another, presumably more active voter. A delegator can override his or her delegate's vote, but would generally choose a delegate based on the delegate's public reputation and likely voting profile.

\subsection{A Key-Encumbrance Dark DAO Design}
\label{sec:impl-design}
Recall the Dark DAO ``guaranteed vote delivery'' property (\Cref{subsec:Dark_DAO_def}): a Dark DAO must guarantee that a bribed voter will cast a vote as specified by the briber. 
Many DAO voting systems, however, including Snapshot, allow voters to change their votes before a proposal has passed. Thus the only way to ensure guaranteed vote delivery is for the Dark DAO to control the voter's voting credential. At the same time, the ``bounded scope'' property of Dark DAOs (\Cref{subsec:Dark_DAO_def}) means that the Dark DAO should have \textit{limited access} to the voter's credential and be able to use it exclusively to cast the vote for which the voter has committed to receiving a bribe. 

We resolve this tension by designing a \textit{key-encumbered wallet}, which stores and manages private keys in smart contracts and enforces access-control policies that we refer to as \textit{encumbrance policies}. A key-encumbered wallet simultaneously allows: (1) use of a key by a Dark DAO specifically for voting in response to a bribe and (2) unrestricted use of the key by its owner for any other purpose.

\paragraph{Key-encumbered wallet.} 
Our key-encumbered wallet 
application is powered by a smart contract that runs on Oasis Sapphire.
The smart contract generates private-public key pairs within a TEE using Sapphire's built-in entropy generation methods. Only the smart contract itself can use the keys it has generated to sign messages.
To create a key-encumbered wallet, one can invoke a ``create wallet'' function in the smart contract using an external account, typically one that is not encumbered. 
We emphasize that while the aforementioned external account is the owner of the wallet, it does not have unfettered access to the private keys created and managed by the wallet smart contract.

Owners of key-encumbered wallets can request signatures for messages by issuing read-only calls to the wallet smart contract. These calls are signed by the owner's external account for authentication.

\paragraph{Dark-DAO encumbrance policy.} 

To create our own Dark DAO based on key-encumbered wallets, we first designed a \textit{key encumbrance policy} contract that regulates all Snapshot-related messages signed by an enrolled wallet, including votes for DAO proposals. The policy will not allow a key owner to sign a vote directly; instead, the owner must unlock the ability to do so for a particular proposal, after which it can sign any voting message related to that proposal. But rather than unlocking a proposal to sign a vote, an owner may delegate its right to vote to a sub-policy: this is the Dark DAO contract. If the vote is given to a sub-policy, the owner forgoes the ability to sign messages relating to that proposal. This mechanism guarantees to the Dark DAO that a user will not change a vote that it signs on the user's behalf. (A user could try to pre-sign a vote, but that would be impractical for, e.g., Snapshot, where ballots incorporate proposal hashes whose inputs include a timestamp, exact proposal title and body, and proposal-text wording. See~\Cref{sec:impl-depl-consider} for more on pre-signing.) The hierarchical design of key encumbrance enables users to participate in several Dark DAOs at once.

\medskip

We summarize the components of our basic Dark DAO prototype in Figure~\ref{fig:dark-dao-pseudocode}, in the form of pseudocode for each of the main functionalities.

\newcommand{\getsr}{{\:{\leftarrow{\hspace*{-3pt}\raisebox{.75pt}{$\scriptscriptstyle\$$}}}\:}}
\newcommand{\sign}{{\mathsf{sign}}}
\newcommand{\keygen}{{\mathsf{keygen}}}

\newcommand{\myInd}{\hspace*{1em}}

\begin{figure}[h!]

\fpage{0.97}{
\underline{Initialization}:
Set $\texttt{accounts} := \{\}, \texttt{bribes} := []$. \\

\underline{On receive $\texttt{keygen()}$ from party $\mathcal{P}$}: \\
\myInd $(\sk, \pk) \getsr S.\keygen()$ \\
\myInd $\texttt{accounts}[\pk] = (\texttt{sk:} ~ \sk, \textsf{party:} ~ \mathcal{P}, \texttt{bribeId:} ~ \bot, \texttt{signed:} ~ \emptyset)$\\
\myInd Send $\pk$ to $\mathcal{P}$. \\

\underline{On receive $\texttt{sign(\pk,}m\texttt{)}$ from party $\mathcal{P}$}: \\
\myInd $\texttt{ret} = (\sk, \mathcal{P}^*, \texttt{bribeId}, \texttt{signed}) \gets \texttt{accounts}[\pk]$\\
\myInd \textsf{assert} $(\texttt{ret} \neq \bot) \wedge (\mathcal{P}^* = \mathcal{P})$\\
\myInd \textsf{assert} $\texttt{bribeId} = \bot \vee m \notin \texttt{bribes[bribeId]}.\mathcal{M}$\\
\myInd $\sigma = S.\sign(\sk, m)$\\
\myInd $\texttt{accounts}[\pk].\texttt{signed.add}(m)$\\
\myInd Send $\sigma$ to $\mathcal{P}$.\\

\underline{On receive $\texttt{registerBribe(bribeAmount}, \mathcal{M})$ from party $\mathcal{B}$ along with $T$ tokens}: \\
\myInd $\texttt{bribeId} \gets \textsf{len}(\texttt{bribes}) + 1$ \\
\myInd $\texttt{bribes}[\texttt{bribeId}] = \texttt{(bribeAmount}, T, \mathcal{M}, \mathcal{B}\texttt{)}$. \\
\myInd Send $\texttt{bribeId}$ to $\mathcal{B}$.\\

\underline{On receive $\texttt{takeBribe(\pk, bribeID)}$ from party $\mathcal{P}$}: \\
\myInd \textsf{assert} $\texttt{accounts}[\pk].\texttt{party} = \mathcal{P}$ \\
\myInd $(\texttt{bribeAmount}, T, \mathcal{M}, \mathcal{B}) \gets \texttt{bribes}[\texttt{bribeId}]$. \\
\myInd \textsf{assert} $\texttt{accounts}[\pk]\texttt{.bribeId} = \bot$ $\wedge$ $\texttt{accounts}[\pk]\texttt{.signed} \cap \mathcal{M} = \emptyset$ $\wedge$ $T \geq \texttt{bribeAmount}$\\
\myInd $\texttt{accounts}[\pk]\texttt{.bribeId} = \texttt{bribeId}$ \\
\myInd $\texttt{bribes}[\texttt{bribeId}].T \;\textnormal{-=}\; \texttt{bribeAmount}$\\
\myInd Send $T$ tokens to $\mathcal{P}$.\\

\underline{On receive $\texttt{signViaEncumberedKey(\pk, 
\textnormal{\textit{m}}, bribeID)}$ from party $\mathcal{B}$}: \\
\myInd $\texttt{bribe} = (\texttt{bribeAmount}, \mathcal{M}, \mathcal{B^*}) \gets \texttt{bribes[accounts[\pk].bribeId]}$ \\
\myInd \textsf{assert} $(\texttt{bribe} \ne \bot) \wedge (\mathcal{B}^* = \mathcal{B}) \wedge m \in \mathcal{M}$ \\
\myInd $\sigma = S.\sign(\sk, m)$\\
\myInd $\texttt{accounts}[\pk].\texttt{signed.add}(m)$\\
\myInd Send $\sigma$ to $\mathcal{B}$.
}
\caption{Key encumbrance and Dark DAO pseudocode}
\label{fig:dark-dao-pseudocode}
\end{figure}

\subsection{Dark DAO execution costs}
\label{sec:dd-cost}
\begin{table}[h!]
\centering
\begin{tabular}{|l|r|c|c|}
\hline
\textbf{Transaction} & \textbf{Gas Usage} & \textbf{ROSE Cost} & \textbf{USD Cost} \\ \hline
Create encumbered account & 237,640 & 0.0237640 & \$0.00123 \\ \hline
Enroll in Snapshot encumbrance policy & 144,981 & 0.0144981 & \$0.00075 \\ \hline
Enter Dark DAO & 167,299 & 0.0167299 & \$0.00086 \\ \hline
Claim bribe payment & 85,064 & 0.0085064 & \$0.00044 \\ \hline
Deploy Snapshot encumbrance policy & 2,543,239 & 0.2543239 & \$0.01314 \\ \hline
Deploy Dark DAO contract & 1,690,955 & 0.1690955 & \$0.00873 \\ \hline
\end{tabular}
\caption{Costs of Dark DAO transactions. \\ 1 ROSE = \$0.05165, as of October 27, 2023. Transactions are priced at 100 Gwei, the Sapphire default.\label{tab:transaction-fees}}
\end{table}

Table~\ref{tab:transaction-fees} describes the costs of the various Oasis Sapphire transactions that are necessary to participate in a Dark DAO. A bribee would perform the first four transactions; among those, the first two are the one-time costs of setting up an encumbered account, and the second two occur whenever a bribe is taken. The Snapshot encumbrance policy deployment is a one-time cost, and Dark DAO contracts would presumably all reference the same Snapshot encumbrance policy until API changes require an upgrade. The Dark DAO contract as written needs to be deployed for every DAO proposal a briber wishes to participate in, but it is straightforward to make the contract reusable.

\subsection{Security} 
\label{subsec:basic_DAO_conf}

We consider an idealized model of Oasis Sapphire's trusted execution environment~\cite{pass2017sgx}, treating side-channel issues and platform-level deployment mistakes (see, e.g.,~\cite{chen2019sgxpectre,van2020sgaxe,jeanlouis2023sgxonerated}) as out of scope for our exploration in this paper. We also assume the integrity, i.e., correct execution, and liveness of Sapphire. Communications with Oasis Sapphire can in principle be observed by a network adversary. The system supports secure channels to application instances, however, and we exclude consideration of side channels resulting from, e.g., analysis correlating Oasis Sapphire traffic with on-chain behavior. (Such side channels can be mitigated through injection of noise, e.g., randomized delays.)

Informally, in this model, the Basic Dark DAO we have described achieves confidentiality (i.e., ``opacity'') as follows. An adversary---an entity seeking to probe the Dark DAO, e.g., on behalf of the target DAO---can mount an active attack against the Dark DAO, posing as a briber and as a set of vote-sellers. Such an adversary can learn two forms of information.  

First, to the extent that it registers bribes, the adversary learns information about voters that accept these bribes. The adversary learns two things about these voters: (1) The number of votes they are selling and (2) Their on-chain addresses. We stress that the adversary learns (1) only for votes it purchases, but those votes no longer then pose a threat to the target DAO. Here, (2) arises in a model where the adversary submits the votes it has obtained via bribery. There are three reasons why (2) is probably of limited practical concern:

\begin{enumerate}
    \item {\em Token fungibility:} If an adversary buys votes from some set of addresses, the adversary controls those votes for a given election. Those addresses might be blacklisted from future participation in the target DAO. But since tokens are fungible, they could simply be sent to new addresses, greatly complicating potential blacklisting policies.
    \item {\em Cost of acquisition:} Buying votes to learn associated addresses---and of course control their votes---is a costly strategy. It also creates a perverse incentive. It actually encourages the creation of Dark DAOs, as the adversary is subsidizing bribes. 
    \item {\em Private voting:} Votes could in principle retain confidentiality during submission, rather than be obtained in cleartext by a briber. A TEE could, for instance, perform submission to a website (e.g., Snapshot). Although Oasis Sapphire does not directly support TLS traffic and thus would not enable straightforward implementation of this functionality, other TEE-based systems can in principle play this role, e.g.,~\cite{zhang2016town}.
\end{enumerate}

The second form of information available to the adversary is the size of bribes registered by (other) bribers---which are published to signal the opportunity to vote sellers. Published bribes only represent an upper bound on Dark DAO activity, however.

In summary, at the application layer, the only feasible way for an adversary to impact the behavior of the Dark DAO through active attack is by buying votes. Additionally, the Dark DAO to the best of our knowledge presents no application-layer denial-of-service attack vectors.

\subsection{Deployment Considerations}
\label{sec:impl-depl-consider}

\paragraph{Pre-signing attacks.} Key-encumbered wallet owners can sign an unlimited number of messages prior to enrolling in an encumbrance policy. This creates an opportunity for wallet owners to pre-sign messages they predict will be encumbered by a policy in the future and defeat the encumbrance scheme. For example, if a wallet owner is about to enroll in an encumbrance policy that restricts its ability to sign a message saying ``vote for Alice,'' the owner could just pre-sign this message before enrolling in the policy. Therefore, encumbrance policies must either be enforced from the moment the wallet is created or work with messages the owner could not have possibly predicted prior to encumbrance, such as those containing a pseudorandom value or the hash of high-entropy inputs, e.g., Snapshot proposal hashes.

An alternative key-encumbered wallet design is to record every message that has ever been signed with an encumbered wallet; on enrollment, an encumbrance policy could then check that no previous message breached the encumbrance assumptions of the policy. However, to assess whether a type of message has ever been signed requires each message to be recorded on the same blockchain as the wallet contract, incurring a cost in transaction fees to the wallet owner whenever he or she wishes to sign a message.

Often, these schemes can be combined to produce workable encumbrance policies that would otherwise require enrollment from the moment of wallet creation. If an encumbrance policy can anchor relevant message characteristics to a specific timestamp (e.g., a DAO proposal's hash to its creation timestamp), it can distinguish whether a particular message could have been signed before the encumbered wallet was enrolled in the policy. All the messages whose timestamps are earlier than the time of enrollment are unrestricted; all those after are eligible for restriction by the policy.

\paragraph{Campaigns.} Rather than selling votes for a given election, a Dark DAO can orchestrate a bribery \textit{campaign}, dispensing bribes that are contingent upon a successful outcome. This might be an election outcome, or the outcome of multiple elections. But a campaign may target any of a range of outcomes reflected in blockchain state---e.g., successful installation of a particular user in a privileged role (e.g., membership in a DAO committee responsible for disbursing funds). Any of a range of bribery policies are also possible, e.g., rewards for recruiting fresh Dark DAO participants.

\subsection{Mitigation: Complete Knowledge}
\label{sec:dd-mitigation}
An application can prevent access to accounts created by a key encumbrance system by requiring a \textit{proof of complete knowledge} to be demonstrated for each public key requesting access~\cite{kelkar2023complete}. Such a proof demonstrates that the associated private key either has been shown or could have been shown, in totality, to an eavesdropper. A key-encumbered wallet can no longer make its encumbrance guarantees if its private key is ever leaked, so correctly implemented key-encumbered wallets are naturally forbidden from creating valid proofs of complete knowledge. Lightweight proofs of complete knowledge via mobile device TEEs may soon be practical if signature verification of the relevant curves becomes cheap, such as through implementation of EIP-7212~\cite{EIP7212}. 



\subsection{Ethical Considerations}
\label{sec:dd-ethical}
We have open-sourced the code of our key encumbered wallet and message-based Dark DAO contracts. We have chosen to do this because, given the current risks of participation in Dark DAOs, the short-term threat is limited. DAOs are currently not highly decentralized, as shown in, e.g.,~\cite{feichtinger2023hidden}---a precondition for Dark DAOs to be effective. We feel it is important to have a clear demonstration of the practicality of Dark DAOs so that the community can understand the contours of the risk in the long term and consider effective countermeasures. 

Additionally, our code has beneficial use cases, which we will show in future work. For example, a confidential DAO whose treasury funds are themselves encumbered inside the wallets of its participants would have sidestepped some of the shortcomings of the Constitution DAO~\cite{Tan:2022}, whose public fundraising appears to have facilitated it being outbid in a silent auction.

\section{Dark DAO Lite}
\label{sec:token-based-dd-impl}

Although our Dark DAO prototype demonstrates that Dark DAOs are practical to build, participating in the Dark DAO as a bribee is not straightforward. The requirements of setting up an encumbered account and managing funds through it, discovering bribes, and enrolling in encumbrance policies create friction for ordinary users.

To alleviate these usability issues and further emphasize the versatility of key encumbrance, we have created a second Dark DAO system that we call a \textit{Dark DAO ``Lite''}. Our Dark DAO Lite scheme involves a trade-off: It achieves greater usability than our basic Dark DAO, but weaker confidentiality. (Thus our use of the term ``lite.'')

The key idea behind the Dark DAO Lite is its use of a DAO-token derivative to hide the complexity of participation. We call this derivative, which is itself a token, a \textit{Dark-DAO token} or \textit{DD token} for short.

DD tokens in a Dark DAO Lite are derived from ordinary tokens in a target DAO through a conversion process. The key steps of this process, which we explain in detail further below, are summarized in Figure~\ref{fig:dark-dao-deposit}. A Dark DAO Lite itself, like a basic Dark DAO, is a smart contract running on Oasis Sapphire and benefits from that chain's confidentiality properties. Its functionality is described in Figure~\ref{fig:tokenized-dd}. DD tokens, however, are ordinary ERC-20 tokens held on Ethereum that can be traded on existing token markets, such as Uniswap. Converting between the target DAO token and the non-voting derivative DD tokens requires technical knowledge, but need only be done once by a small set of actors---who can obtain remuneration through DD token markets. Once DD tokens are created, no sophistication is required to manage them. 

DD tokens realize a concept that we refer to as \textit{DAO-token fractionalization}.

\paragraph{DAO-token fractionalization.} DAO tokens grant two capabilities to their owner: (1) the ability to sign votes on proposals and (2) ownership rights in the DAO, expressible as the ability to send the tokens to a different address. Key-encumbered wallets and encumbrance policies enable a separation of these two capabilities by creating two paths of access control to a single private key controlling the tokens of a target DAO. This fractionalization yields two distinct resources: 

\begin{itemize}
    \item \textit{Voting rights} corresponding to converted target-DAO tokens. A pool of these voting rights may be  purchased by a vote-buyer / briber through an auction mechanism we describe below. We refer to the pool of voting rights as \textit{self-auctioning}, since the auction process is automated and requires no intervention by DD-token holders.
    \item \textit{DD tokens}, which may be individually owned and correspond to ownership rights in the target DAO plus the right to receive revenue from the auctioning of the pool of fractionalized voting rights.
\end{itemize}

In the remainder of this section, we explain how the various parts of a Dark DAO Lite work (Section~\ref{subsec:tokenized_DD_mechanics}) and its security properties (Section~\ref{subsec:tokenized_DD_security}).  In what follows, unless otherwise specified, we use the term Dark DAO to refer to a Dark DAO Lite. 

\subsection{Dark-DAO Lite functionality}
\label{subsec:tokenized_DD_mechanics}


\definecolor{myred}{RGB}{208, 2, 27}

\tikzset{every picture/.style={line width=0.75pt}} 

\begin{figure*}[t]
    \centering
\scalebox{0.9}{
\begin{tikzpicture}[x=0.75pt,y=0.75pt,yscale=-1,xscale=1]


\node[circle,fill,minimum size=5mm] (head) at (320,64) {};
\node[rounded corners=2pt,minimum height=1.3cm,minimum width=0.4cm,fill,below = 1pt of head] (body) {};
\draw[line width=1mm,round cap-round cap] ([shift={(2pt,1pt)}]body.north east) --++(90:6mm);
\draw[line width=1mm,round cap-round cap] ([shift={(-2pt,1pt)}]body.north west)--++(90:6mm);
\draw[thick,white,-round cap] (body.south) --++(90:5.5mm);

\draw  [color={rgb, 255:red, 208; green, 2; blue, 27 }  ,draw opacity=1 ] (457,78.4) .. controls (457,72.66) and (461.66,68) .. (467.4,68) -- (509.6,68) .. controls (515.34,68) and (520,72.66) .. (520,78.4) -- (520,109.6) .. controls (520,115.34) and (515.34,120) .. (509.6,120) -- (467.4,120) .. controls (461.66,120) and (457,115.34) .. (457,109.6) -- cycle ;
\draw    (320,127) -- (320,195) ;
\draw [shift={(320,197)}, rotate = 270] [color={rgb, 255:red, 0; green, 0; blue, 0 }  ][line width=0.75]    (10.93,-3.29) .. controls (6.95,-1.4) and (3.31,-0.3) .. (0,0) .. controls (3.31,0.3) and (6.95,1.4) .. (10.93,3.29)   ;
\draw  [color={rgb, 255:red, 0; green, 0; blue, 0 }  ,draw opacity=1 ] (89,81.4) .. controls (89,75.66) and (93.66,71) .. (99.4,71) -- (163.6,71) .. controls (169.34,71) and (174,75.66) .. (174,81.4) -- (174,112.6) .. controls (174,118.34) and (169.34,123) .. (163.6,123) -- (99.4,123) .. controls (93.66,123) and (89,118.34) .. (89,112.6) -- cycle ;
\draw [color={rgb, 255:red, 208; green, 2; blue, 27 }  ,draw opacity=1 ]   (340,89) -- (449,89) ;
\draw [shift={(451,89)}, rotate = 180] [color={rgb, 255:red, 208; green, 2; blue, 27 }  ,draw opacity=1 ][line width=0.75]    (10.93,-3.29) .. controls (6.95,-1.4) and (3.31,-0.3) .. (0,0) .. controls (3.31,0.3) and (6.95,1.4) .. (10.93,3.29)   ;
\draw   (325.5,161) .. controls (325.5,157.13) and (328.63,154) .. (332.5,154) .. controls (336.37,154) and (339.5,157.13) .. (339.5,161) .. controls (339.5,164.87) and (336.37,168) .. (332.5,168) .. controls (328.63,168) and (325.5,164.87) .. (325.5,161) -- cycle ;
\draw  [color={rgb, 255:red, 208; green, 2; blue, 27 }  ,draw opacity=1 ] (347.5,65) .. controls (347.5,61.13) and (350.63,58) .. (354.5,58) .. controls (358.37,58) and (361.5,61.13) .. (361.5,65) .. controls (361.5,68.87) and (358.37,72) .. (354.5,72) .. controls (350.63,72) and (347.5,68.87) .. (347.5,65) -- cycle ;
\draw [color={rgb, 255:red, 208; green, 2; blue, 27 }  ,draw opacity=1 ]   (448.2,102.2) -- (342,102.99) ;
\draw [shift={(340,103)}, rotate = 359.58] [color={rgb, 255:red, 208; green, 2; blue, 27 }  ,draw opacity=1 ][line width=0.75]    (10.93,-3.29) .. controls (6.95,-1.4) and (3.31,-0.3) .. (0,0) .. controls (3.31,0.3) and (6.95,1.4) .. (10.93,3.29)   ;
\draw  [color={rgb, 255:red, 208; green, 2; blue, 27 }  ,draw opacity=1 ] (355.5,115) .. controls (355.5,111.13) and (358.63,108) .. (362.5,108) .. controls (366.37,108) and (369.5,111.13) .. (369.5,115) .. controls (369.5,118.87) and (366.37,122) .. (362.5,122) .. controls (358.63,122) and (355.5,118.87) .. (355.5,115) -- cycle ;
\draw    (297,90) -- (181.2,90.2) ;
\draw [shift={(179.2,90.2)}, rotate = 359.9] [color={rgb, 255:red, 0; green, 0; blue, 0 }  ][line width=0.75]    (10.93,-3.29) .. controls (6.95,-1.4) and (3.31,-0.3) .. (0,0) .. controls (3.31,0.3) and (6.95,1.4) .. (10.93,3.29)   ;
\draw   (181.5,72) .. controls (181.5,68.13) and (184.63,65) .. (188.5,65) .. controls (192.37,65) and (195.5,68.13) .. (195.5,72) .. controls (195.5,75.87) and (192.37,79) .. (188.5,79) .. controls (184.63,79) and (181.5,75.87) .. (181.5,72) -- cycle ;
\draw    (184,102) -- (296,102.98) ;
\draw [shift={(298,103)}, rotate = 180.5] [color={rgb, 255:red, 0; green, 0; blue, 0 }  ][line width=0.75]    (10.93,-3.29) .. controls (6.95,-1.4) and (3.31,-0.3) .. (0,0) .. controls (3.31,0.3) and (6.95,1.4) .. (10.93,3.29)   ;
\draw   (187.5,114) .. controls (187.5,110.13) and (190.63,107) .. (194.5,107) .. controls (198.37,107) and (201.5,110.13) .. (201.5,114) .. controls (201.5,117.87) and (198.37,121) .. (194.5,121) .. controls (190.63,121) and (187.5,117.87) .. (187.5,114) -- cycle ;

\draw (453,80.2) node [anchor=north west][inner sep=0.75pt]  [font=\footnotesize,color={rgb, 255:red, 208; green, 2; blue, 27 }  ,opacity=1 ] [align=left] {\begin{minipage}[lt]{50.64pt}\setlength\topsep{0pt}
\begin{center}
Dark DAO\\contract
\end{center}

\end{minipage}};
\draw (94,82) node [anchor=north west][inner sep=0.75pt]  [font=\footnotesize,color={rgb, 255:red, 0; green, 0; blue, 0 }  ,opacity=1 ] [align=left] {\begin{minipage}[lt]{53.98pt}\setlength\topsep{0pt}
\begin{center}
DD token\\ contract
\end{center}

\end{minipage}};
\draw (278,224) node [anchor=north west][inner sep=0.75pt]   [align=left] {0x...};
\draw (336,223) node [anchor=north west][inner sep=0.75pt]   [align=left] {0x...};
\draw (305,205) node [anchor=north west][inner sep=0.75pt]   [align=left] {0x...};
\draw (327,155) node [anchor=north west][inner sep=0.75pt]  [font=\scriptsize] [align=left] {1. \ Transfer DAO token to \\ \ \ \ \ \ encumbered account};
\draw (349,59) node [color={rgb, 255:red, 208; green, 2; blue, 27 }  ,draw opacity=1 ] [anchor=north west][inner sep=0.75pt]  [font=\scriptsize] [align=left] {2. \ Register\\ \ \ \ \ \ proof of deposit};
\draw (357,109) node [color={rgb, 255:red, 208; green, 2; blue, 27 }  ,draw opacity=1 ] [anchor=north west][inner sep=0.75pt]  [font=\scriptsize] [align=left] {3. \ Grant mint \\ \ \ \ \ authorization};
\draw (183,66) node [anchor=north west][inner sep=0.75pt]  [font=\scriptsize] [align=left] {4. \ Forward\\ \ \ \ \ \ mint authorization};
\draw (189,108) node [anchor=north west][inner sep=0.75pt]  [font=\scriptsize] [align=left] {5. \ Receive DD token};
\draw (290,39) node [anchor=north west][inner sep=0.75pt]  [font=\footnotesize] [align=left] {Depositor};

\end{tikzpicture}}
\caption{Transaction flow for converting DAO tokens to DD tokens. The Dark DAO contract, denoted in \textcolor{myred}{red}, is on Oasis, while the rest is on Ethereum. The accounts at the bottom are encumbered, under the control of the Dark DAO contract; each deposit is sent to a fresh encumbered account.}
\vspace{-0.3cm}
\label{fig:dark-dao-deposit}
\end{figure*}
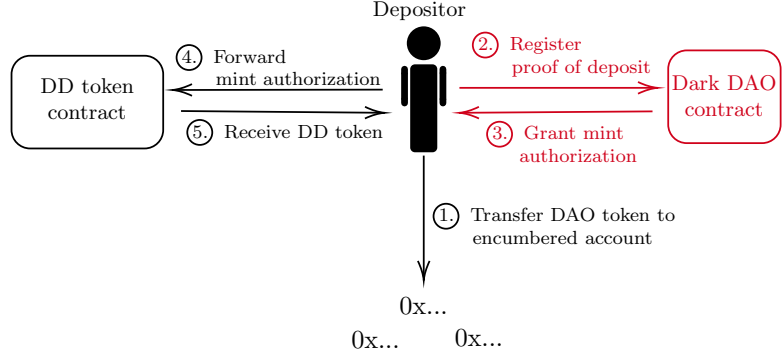
\newcommand{\functioncode}{\textbf{function }}
\newcommand{\forcode}{\textbf{for }}
\newcommand{\returncode}{\textbf{return }}
\newcommand{\myind}{\hspace*{1em}}





\begin{figure}[h!]
    \fpage{0.9}{
        \begin{center}
            \textbf{DD Token, extending ERC-20}
        \end{center}
        \underline{Initialize(\textsf{pk})}: $\textsf{DDpk}:= \textsf{pk}$, $\textsf{supply} := 0$, $\textsf{authNonces} := \{\}$  \\

        \underline{On receive $\texttt{mint}(m = (T, \textsf{nonce}), \sigma)$ from party $\player$}: \\
        \myInd \textsf{assert} $S.\textsf{ver}(\textsf{DDpk}, m, \sigma)$ \\
        \myInd \textsf{assert} $\textsf{nonce} \notin \textsf{authNonces}[\player]$ \\
        \myInd $\textsf{supply} \gets \textsf{supply} + T$\\
        \myInd $\textsf{authNonces}[\player].\textsf{add}(\textsf{nonce})$ \\
        \myInd Send $T$ tokens to $\player$ \\

        \underline{On receive $\texttt{burn}()$ from party $\player$ along with $T$ DD-tokens}: \\
        \myInd $\textsf{supply} \gets \textsf{supply} - T$
    
        }
    \caption{DD token pseudocode}
    \label{fig:dd-token}
\end{figure}

\begin{figure}[h!]
    \fpage{0.9}{
        \begin{center}
            \textbf{Dark DAO Lite}
        \end{center}

        $\functioncode \texttt{initialize}(\textrm{header}):$\\
        \myind $\mathsf{eth\_block\_header} := \textrm{header}$ \hspace{2em}// updated by an oracle or by piggybacking on proofs\\
        \myind $\mathsf{encumbered\_accounts} := \{\}$ \\
        \myind $\mathsf{dark\_dao\_sk,dark\_dao\_pk} \getsr \mathsf{S}.\texttt{keygen()}$ \\
        \myind $\mathsf{balances} := \{\}$ \\
        \myind $\mathsf{registered\_proofs} := \{\}$ \\
        \myind \returncode $\mathsf{dark\_dao\_pk}$\\

        $\functioncode \texttt{get\_deposit\_address}():$\\
        \myind $\mathsf{sk,pk} \getsr \mathsf{S}.\texttt{keygen()}$ \\
        \myind $\mathsf{encrypted\_key} \getsr \mathsf{dark\_dao\_sk}.\texttt{encrypt(sk)}$\\
        \myind \returncode $\mathsf{pk}, \mathsf{encrypted\_key}$\\

        $\functioncode \texttt{deposit\_and\_mint}(\pi, \textrm{recipient}):$\\
        \myind assert $\pi \notin \mathsf{registered\_proofs}$\\
        \myind assert $\texttt{verify\_deposit\_proof}(\mathsf{eth\_block\_header},\pi)$ \\
        \myind $\mathsf{encumbered\_accounts}.\texttt{insert}(\pi.\textrm{pk}, \mathsf{dark\_dao\_sk}.\texttt{decrypt}(\pi.\textrm{encrypted\_key}))$\\
        \myind $\mathsf{balances}[\pi.\textrm{pk}] := \mathsf{balances}[\pi.\textrm{pk}] + \pi.\textrm{amount}$\\
        \myind $\mathsf{registered\_proofs}.\texttt{insert}(\pi)$\\
        \myind $\mathsf{message.amount} := \pi.\textrm{amount}$\\
        \myind $\mathsf{message.recipient} := \pi.\textrm{recipient}$\\
        \myind \returncode $\mathsf{message, dark\_dao\_sk.\texttt{sign}(message)}$\\

        $\functioncode \texttt{redeem\_and\_withdraw}(\pi, \textrm{recipient}):$\\
        \myind assert $\pi \notin \mathsf{registered\_proofs}$\\
        \myind assert $\texttt{verify\_burn\_proof}(\mathsf{eth\_block\_header},\pi)$\\
        \myind $\mathsf{registered\_proofs}.\texttt{insert}(\pi)$\\
        \myind $\mathsf{accounts, amounts} := \texttt{select\_withdrawal\_accounts}(\mathsf{encumbered\_accounts, balances}, \pi.\textrm{amount})$\\
        \myind $\mathsf{signed\_transactions} := \{\}$\\
        \myind $\forcode \mathsf{account, amount} \in \mathsf{accounts, amounts}$\\
        \myind \myind $\mathsf{signed\_transactions}.\texttt{insert}(\mathsf{account.sk.}\texttt{sign}(\texttt{transfer\_from}(\mathsf{account.pk, amount}, \texttt{recipient})))$\\
        \myind \myind $\mathsf{balances[account.pk] := balances[account.pk] - amount}$\\
        \myind \returncode $\mathsf{signed\_transactions}$\\
        }
    \caption{Pseudocode for Dark DAO Lite smart contract on Oasis}
    \label{fig:tokenized-dd}
\end{figure}    

\paragraph{Converting target-DAO tokens to DD tokens.} Before authorizing the creation of new DD tokens, the Dark DAO needs to gain control over target-DAO tokens. This is accomplished by having the target-DAO tokens sent to a freshly generated Ethereum account under the Dark DAO's control. 

A user wishing to receive DD tokens queries the Dark DAO in an off-chain query for a newly generated such address into which a batch of target-DAO tokens may be deposited. The Dark DAO contract responds with the deposit address $A$ and a ciphertext $C$ on deposit data: ($\sk_A$, $R$), where $\sk_A$ is the private key for $A$ and $R$ is the address for receiving DD tokens minted as a result of a deposit to $A$.

The reason for creating a fresh address for each deposit is \textit{confidentiality} of these addresses. Because the generation process happens off chain, there is no public indication that $A$ is controlled by the Dark DAO. The deposit is indistinguishable from a simple transfer of target-DAO tokens to a new EOA (externally owned account).

After the user transfers target-DAO tokens to $A$, the user submits a state proof of $A$'s target-DAO token balance to the Dark DAO contract, along with the ciphertext $C$. The Dark DAO contract checks the proof which begins an optional \textit{lockup period} on the tokens. (See Section~\ref{subsec:tokenized_DD_security} for details.) After the lockup period is complete, the user can query it to receive a signed message authorizing the minting of an equivalent amount of DD tokens. A user with can then send this signed message to the DD token contract to mint the DD tokens, as shown in Figure~\ref{fig:dd-token}. The mint operation need not happen immediately; the user could presumably wait until the DD tokens need to be transferred or sold.

The use of state proofs serves as a bridge between Ethereum and Oasis. We assume a trusted source of block hashes for this purpose. The ``Oasis Privacy Layer,'' which uses the Celer Network as a bridge system under the hood, is the existing bridge from Oasis Sapphire to any supported EVM network.


\paragraph{Redeeming DD tokens for target-DAO tokens.}

The conversion process may be reversed: DD tokens can be redeemed for their underlying target-DAO tokens. A user holding DD tokens first issues a burn transaction of $n$ tokens to the DD token contract, which removes the $n$ DD tokens from circulation and records a receipt of the burn to persistent storage. The user then submits a state proof of the burn receipt to the Dark DAO contract on Oasis, which in return sends back a proportional amount of bribe money and authorizes the user to submit off-chain withdrawal requests to the Dark DAO contract. The Dark DAO responds to these requests with a signed Ethereum transaction which transfers up to $n$ target-DAO tokens from a Dark DAO controlled account to the user. It is the user's responsibility to include this transaction on the Ethereum mainnet. Note that DD tokens must be fungible and liquid, or else they would not be easily tradable; therefore, the Dark DAO contract must be able to handle partial withdrawals from its accounts. A withdrawal that is greater than the current withdrawal account's balance would require multiple withdrawal transactions.

On Ethereum, transactions are ordered by sender according to increasing transaction nonce: the first transaction by a particular sender must be signed with nonce 0, the second with nonce 1, and so on. Target-DAO token transfers out of Dark DAO accounts are also transactions and must be included in increasing nonce order. To prevent users who fail to include their transactions on the Ethereum mainnet from blocking other target-DAO token withdrawals, everyone who is ready to withdraw is issued a signed transaction from the same Dark DAO account and with the same nonce. The first withdrawal transaction to be included in an Ethereum block ``wins,'' and the other competing transactions with the same nonce and sender are automatically invalidated at no cost, per Ethereum's rules. To allow the next withdrawal to process, a user can show a Merkle proof of transaction inclusion in an Ethereum block, which simultaneously increments the nonce of the Dark DAO account (or chooses a new withdrawal account) and marks the included withdrawal as completed. 

Ethereum transactions need to be funded before they are included, so to pay for the target-DAO token transfer, some ETH must be sent to the Dark DAO account in an earlier transaction. We expect withdrawers will use Flashbots bundles to execute the funding and token transfer transactions atomically and to prevent other withdrawals from backrunning their funding transactions.


\paragraph{DD tokens.} 
As we have explained, DD tokens are the primary financial instrument of a Dark DAO Lite, issued when deposits of target DAO tokens are made to Dark DAO accounts. They can later be redeemed for the underlying target DAO tokens plus any accumulated bribes on the voting rights to the encumbered target DAO tokens. The Dark DAO smart contract on Oasis acts as the primary controller of all participating encumbered Ethereum accounts and is itself an encumbrance policy. 

We emphasize that users with DD tokens cannot vote in the target DAO with them; this voting ability is in the Dark DAO's self-auctioning pool. Rather, users with DD tokens hold a claim to ownership in the target DAO plus a proportional fraction of the bribe revenue the Dark DAO creates from selling its votes. In short, 1 DD token is equivalent to the ownership rights of 1 target DAO token plus fractional bribe revenue. 

\paragraph{Voting-rights auctioning.}
We assume that the target DAO utilizes a message-based, off-chain voting system with voting power assigned to accounts based on their DAO token balances, though the Dark DAO contract could be adapted for other voting schemes.

When a target-DAO proposal is published, the proposal hash is made public. Bribers who wish to purchase the Dark DAO's voting power for the proposal can start an auction for that proposal from the Dark DAO contract and bid on the right to sign votes from all Dark DAO accounts\footnote{In our implementation, bids are made in Oasis's native token, ROSE. Over time, the DD token will have increased exposure to this other asset. To remedy this, the Dark DAO contract can sell its proceeds periodically for the target-DAO token.}. We implemented a first-price auction in our implementation, but other auction types could be substituted. All auctions have a fixed duration and must end before the proposal expires, or else the purchased votes cannot be used. 

A briber who wins an auction can ask the Dark DAO contract to sign votes for the proposal from all of the Dark DAO accounts. Bidders could bid on a hash that does not correspond to a proposal, so as implemented, any auction winner can enumerate Dark DAO accounts by reading the vote signatures. In Section~\ref{subsec:tokenized_DD_security}, we briefly discuss an alternative, privacy-preserving approach.

\paragraph{How the DD-token market works.}  
As conversion of target-DAO tokens to DD tokens requires a (small) degree of technical knowledge, including interaction with the Oasis Sapphire chain, our expectation is that arbitrageurs will perform the conversion and sell DD tokens in exchanges for Ethereum. As ordinary ERC-20 tokens, DD tokens may be sold in either decentralized or centralized exchanges. The value of generating DD tokens---and thus revenue for arbitrageurs---may be priced into the market value of DD tokens.

Whether a given user $\player$ prefers to hold target-DAO tokens or DD tokens depends upon the value to the user of voting, which is related $\util_{\player}(E,\true)$ for a set of elections $E$ over which the user intends to hold DD tokens. Given high utility, i.e., a particular desired outcome, a user may prefer to vote and thus hold target-DAO tokens. Many users, however, are apathetic (as discussed in~\Cref{{subsec:apathy}}) and would derive higher utility instead from holding DD tokens. The technical requirements and user experience for holding the two types of token are identical.

We emphasize that DD tokens may be redeemed for target-DAO tokens. Hence the fair market price of DD tokens should be at least that of target-DAO tokens minus the transaction cost for redemption. 

\bigskip

Aside from making Dark DAOs more practical, our DD-token scheme demonstrates a concept of broad interest: key encumbrance enables new financial assets with sophisticated policies to be created from the restructuring of existing ones. While use of TEEs to realize this concept has been previously explored~\cite{matetic2018delegatee,puddu2019teevil}, our work is the first instance of which we're aware in which such assets are realized as decentralized-finance tokens.

\paragraph{Execution costs.} Table \ref{tab:transaction-fees-2} outlines the transaction costs of creating and participating in a Dark DAO Lite.

\begin{table}[h!]
\centering
\begin{tabular}{|l|r|c|r|}
\hline
\textbf{Ethereum Transaction} & \textbf{Gas Usage} & \textbf{ETH Cost} & \textbf{USD Cost} \\ \hline
Deploy DD token contract (\textit{one-time cost}) & 1,552,447 & 0.0240629 & \$42.88543 \\ \hline
Transfer DAO token to Dark DAO account & 51,438 & 0.0007973 & \$1.42096 \\ \hline
Mint DD tokens & 99,050 & 0.0015353 & \$2.73624 \\ \hline
Burn DD tokens & 58,665 & 0.0009093 & \$1.62057 \\ \hline
Fund DAO token transfer from Dark DAO account & 21,000 & 0.0003255 & \$0.58011 \\ \hline
DAO token transfer from Dark DAO account & 51,438 & 0.0007973 & \$1.42096 \\ \hline

\textbf{Oasis Transaction} & \textbf{Gas Usage} & \textbf{ROSE Cost} & \textbf{USD Cost} \\ \hline
Deploy Dark DAO contracts (\textit{one-time cost}) & 6,801,449 & 0.6801449 & \$0.03513 \\ \hline
Prove DAO token deposit to Dark DAO contract & 863,885 & 0.0863885 & \$0.00446 \\ \hline
Prove DD token burn to Dark DAO contract & 501,138 & 0.0501138 & \$0.00259 \\ \hline
Prove DD withdrawal inclusion & 332,495 & 0.0332495 &  \$0.00171 \\ \hline
Create DAO proposal voting rights auction & 122,912 & 0.0122912 & \$0.00063 \\ \hline
Bid on voting rights for a proposal & 53,017 & 0.0053017 & \$0.00027 \\ \hline
\end{tabular}

\caption{Costs of Dark DAO Lite transactions. \\ 1 ETH = \$1,782.22, as of October 27, 2023. Ethereum transactions are priced at 13.5 Gwei, the 60-day average ending on October 26, 2023. \\ 1 ROSE = \$0.05165, as of October 27, 2023. Oasis transactions are priced at 100 Gwei, the Sapphire default. \label{tab:transaction-fees-2}}
\end{table}

\subsection{Security}
\label{subsec:tokenized_DD_security}

We assume the same security model as in Section~\ref{subsec:basic_DAO_conf}. 
A Dark DAO Lite, as noted above, achieves \textit{weaker confidentiality} than a basic Dark DAO. (Although its integrity and DoS properties are the same.) 

The main reason the Dark DAO Lite does not achieve the same strong confidentiality as the basic Dark DAO is the \textit{liquidity} of DD tokens. Recall that when \textit{deposited}, target-DAO tokens are transferred to a Dark-DAO-generated address that is indistinguishable on chain from an ordinary EOA. When DD tokens are \textit{redeemed} for target-DAO tokens, however, the address in which they are held is revealed to the redeeming player to be part of the Dark DAO. Furthermore, the liquidity of tokens means that one user can redeem tokens deposited by another user. An adversary can perform redemptions strategically in an attempt to enumerate Dark-DAO deposit addresses. 

For example, suppose that players only deposit one token at a time and that redemptions are LIFO (last in, first out). An adversary $\adv$ can deposit a token and wait until another player $\player_1$ deposits a token. $\adv$ then withdraws its token, revealing $\player_1$'s deposit address, and then redeposits its token. $\adv$ can do the same when $\player_2$ deposits a token, and so forth. In principle, by withdrawing and depositing $t$ tokens over intervals of length $\Delta$, where $t$ is at least as large as other players' deposits over any interval of length $\Delta$, $\adv$ can identify all deposit addresses. 

Such discovery of deposit addresses through strategic redemption is somewhat costly in practice, because it incurs transaction costs. Our current Dark DAO Lite implementation in fact uses FIFO scheduling. It is possible to impede adversarial address-discovery strategies against FIFO by imposing a lockup period on deposited tokens. The effect of this practice is to raise the adversary's capital requirements, i.e., require the adversary to control a large number tokens. Other approaches might be more effective and their exploration constitutes an interesting research challenge.

An additional form of information leakage in a Dark DAO Lite arises because the circulating supply of DD tokens is publicly visible. The quantity of these tokens specifies corresponds to the size of the available pool of votes available for purchase in the Dark DAO. (It is possible in principle to enhance a Dark DAO Lite to mint fake DD tokens, a potential future enhancement.) Furthermore, on-chain transaction analysis may leak futher information. For example, the timing of target-DAO token deposits and DD token issuance can help an adversary infer target-DAO token addresses.

\section{Summary Guidance for DAOs}
\label{sec:guidance}


Our \indexshort framework and the implications we show in~\Cref{sec:in-practice} suggest a number of forms of concrete guidance for DAOs seeking to enforce or improve meaningful decentralization.  We discuss them in this section. We summarize our guidance for practitioners  in~\Cref{tab:recommendations}.

\begin{table}[h!]
\resizebox{\columnwidth}{!}{
\begin{tabular}{|p{4cm}|p{5cm}|p{6cm}|p{3cm}|}
  \hline     \textit{Topic}  & \textit{General Guidance} & \textit{Reason}  & \textit{Relevant result}
       \\
        \hline &&&\\
     1. \textbf{Vote delegation} & Given a large inactivity whale, vote delegation tends to increase decentralization.  & Delegation (counterintuitively) increases decentralization  by diversifying tokens away from a big inactivity whale.  & Thm.~\ref{thm:delegation}\\
        \hline &&&\\
     2. \textbf{Voting privacy} & Voting privacy increases decentralization. & Private voting eases herding, whose effects are centralizing.  & Thm.~\ref{thm:herding}\\ 
             \hline &&&\\
      3. \textbf{Voter bribery} & The scale of bribery increases with decentralization. & Low alignment of utility functions means systemic coordination is required to impose alignment.   & Thms.~\ref{thm:bribery1},~\ref{thm:bribery2}, and~\ref{thm:bribery3}\\
              \hline &&&\\
    4. \textbf{Dark DAO risks} & Dark DAO risks are likely to increase with decentralization. & As bribery coordination costs grow, Dark DAOs become a more compelling approach to influencing vote outcomes.   & Inference from Thm.~\ref{thm:bribery2} and~\ref{thm:bribery3}\\ 
            \hline &&&\\
                5. \textbf{Dark DAO feasibility} & Dark DAOs are feasible today. & We have shown that existing tools enable effective Dark-DAO deployment. Technical feasibility is unlikely to prove a barrier to their use by adversaries. Complete Knowledge (CK) for voter keys may be a useful countermeasure.    & Sections~\ref{sec:Dark_DAO_implementation} and~\ref{sec:token-based-dd-impl}\\ 
            \hline &&&\\
    6. \textbf{Identity verification} & Weak identity verification increases centralization in quadratic voting. &  A whale that can spread tokens across identities amplifies its voting power.  & Analysis in Sections~\ref{subsec:quad-voting} and~\ref{subsec:Dark_DAO_goals}\\ 
                 \hline &&& \\
   7. \textbf{Voting slates / proposal bundling} &  Bundling choices into slates (like protocol upgrades that include many voting issues in one package) decreases decentralization. & Bundled choices artificially align otherwise heterogeneous utility functions and/or induce apathy by smoothing out utility functions.  & Thm.~\ref{thm:voting-slates}\\ 
             \hline &&&\\
   8. \textbf{Data collection} & Careful voting-statistic collection facilitates decentralization measurement. & Lack of systematic collection and publication of detailed voting statistics makes decentralization measurement challenging today.     & Discussion in~\Cref{sec:vbe-practical-considerations}.\\ 
                \hline 
    \end{tabular}
}
    \captionof{table}{Guidance implied by this paper's results regarding DAO decentralization.}
     \label{tab:recommendations}
\end{table}


\paragraph{Apathy / inactivity whale and delegation:} As we show in~\Cref{subsec:apathy}, token holders who do not vote---those, in a rational model, with near-zero utility functions---have a centralizing effect. Recall that our term for this group is the \textit{inactivity whale}.

One way to diminish the size of the inactivity whale is through delegation. Intuitively, if tokens associated with the inactivity whale are distributed between at least two delegatees in distinct clusters, then they come to represent distinct utility functions—--and thus contribute to decentralization. 

We show in~\Cref{subsec:delegation} that when the inactivity whale is large---with respect to delegatees---delegation increases decentralization. (Otherwise, delegation may or may not have this effect.)

\paragraph{Herding / voting privacy:} There is anecdotal evidence suggesting that social pressure causes herding---specifically that voters align themselves with whales or voting blocs~\cite{sharma2023unpacking}. We may view the effect as a shift in utility function. As we show in~\Cref{subsec:herding}, this shift has a centralizing effect.

Herding arises because votes are publicly visible. Voting privacy in principle alleviates such pressure and therefore has a decentralizing effect. 

Snapshot, a popular platform for DAO voting, has recently implemented a form of privacy called \textit{shielded voting}~\cite{Snapshot:2022}.This form of privacy, however, is only ephemeral: Votes are private when submitted, but revealed at the end of the vote-casting period. So it is unclear that it can fully address the centralizing effects of herding.

End-to-end verifiable voting systems have been proposed in the literature for decades that achieve both voting integrity and confidentiality~\cite{ali2016overview}. How to implemented them with token-based weighting is, to the best of our knowledge, though, an open problem. 

\paragraph{Voter bribery:} Our work shows a relationship between centralization and bribery. In general, bribery causes an increase in centralization, as it has the effect of aligning the utility functions of other players with those of the briber, as we show in~\Cref{subsec:vbe-and-bribery}. 

We also show that as decentralization increases, bribery cost increases. Roughly speaking, increasing diversity in utility functions means increasing cost to align them. 

DAOs today are largely centralized~\cite{sharma2023unpacking, feichtinger2023hidden, bribery, vitalikgov,daodec}. Bribery may not be especially useful, as whales generally exert strong control and require relatively little coordination to align utility functions into a favorable voting bloc. Voter bribery, however, is a problem in many settings, both in political voting~\cite{mcgee2023often} and in corporate governance (see, e.g.,~\cite{Schickler:2023}).

One implication of our results is that as DAO decentralization increases, in order for bribery to succeed, it will need to be systemic. DAO designers should therefore recognize large-scale bribery as a future risk.

\paragraph{Dark DAO risks:} We hypothesize that the most technically feasible way to implement large-scale bribery is through a Dark DAO. 

We have presented in~\Cref{sec:Dark_DAO_implementation} the first fully functional private Dark DAO capable of subverting votes on Ethereum. Our architecture leverages the privacy assurances of TEEs in Oasis, but bridges to Ethereum, where most DAOs operate. Our results show that Dark DAOs are technically feasible (and incur low transaction costs) and thus represent a viable future threat. 

Dark DAOs pose not just a technical threat, but also a psychological one. The mere existence of a Dark DAO may create a perception of vote-manipulation even if the Dark DAO has minimal impact. Moreover, Dark DAOs can be used not just for direct bribery but also for more subtle attacks. They can, for instance, subvert quadratic voting schemes even when such schemes rely on well-functioning decentralized identity systems. 

One possible countermeasure DAO designers may ultimately wish to consider is requiring voting participants to execute complete-knowledge (CK) proofs on their keys~\cite{kelkar2023complete}.

\paragraph{Voting slates / bundling proposals:} A common trick for passing pieces of legislation that are unpopular or have a narrow base of support is to bundle them together in large, bills. Earmarks are a prime example~\cite{Earmarks:2023}.

This practice may be regarded as a form of utility-function ``smoothing'': The utility function of the bill as a whole (for the legislators voting on it) differs from that of its components.

As the practice of bundling proposals / measures has the goal of aligning utility functions, 
from the standpoint of \indexshort, it generally has a centralizing effect, as we show in Section~\ref{subsec:voting-slates}. DAOs may therefore wish to consider limiting the practice and instead explore way to unbundle multi-component proposals.

\paragraph{Data collection:} There is no practical way to compute \indexshort directly---since players' typically do not express their utility functions. As we discuss in~\Cref{sec:vbe-practical-considerations}, however, there are ways to estimate it for a DAO based on voting history. We have found it challenging to collect full voting histories for even popular DAOs. A recommendation for the community is to establish and adhere to standards for archival preservation of DAO voting data. A few

\silence{\section{Related Work}
\label{sec:related}

\paragraph{DAOs:} Research literature on DAOs has been limited to date, but fairly broad. It has included measurement studies~\cite{feichtinger2023hidden,sharma2023unpacking}, retrospectives on the failure of The DAO (e.g., \cite{dupont2017experiments}) and ways of addressing related technical flaws in smart contracts such as dangerous reentrancy (e.g.,~\cite{luu2016making,cecchetti2021compositional}), DAO mechanism design (e.g.,~\cite{bahrani2023bidders}), and exploration of DAOs from the standpoint of legal theory (e.g.,~\cite{hassan2021decentralized,werbach2018trust}) and economics and governance (e.g.,~\cite{beck2018governance}). 

Works exploring measurement of DAOs' degree of decentralization most notably include Feichtinger et al.~\cite{feichtinger2023hidden},who explore Gini and Nakamoto indices, as well as participation rates and the monetary cost of governance, Sharma et al.~\cite{sharma2023unpacking}, who consider various notions of entropy, as well as participation rates and graph-based measures of decentralization, and~\cite{wright2021measuring}, which taxonomizes DAOs by comparison with other autonomous systems. Sun et al.~use clustering to identify voting blocs in a study of MakerDAO~\cite{sun2022multiparty}. Also of note is the informal notion of ``credible neutrality,'' a community standard articulated in, e.g.,~\cite{Buterin:2020,Buterin:2023}.

\paragraph{Social choice and voting theory:}
A long line of work on social choice and voting theory investigates how best to aggregate preferences of individual voters---the same functionality that DAOs seek to provide in the decentralized setting. There are some major differences in the DAO setting however, which may reduce how effective existing techniques will be. For instance, the permissionless nature of DAOs allows for the presence of Sybils which is not typically accounted for in existing voting theory literature. Further, while the threat of large-scale voter bribery is typically safe to ignore in classical voting, both due to the high likelihood of detecting such an attack, as well as the the challenge in coordinating the attack itself, as shown in our paper, Dark DAOs invalidate these prior assumptions in the DAO setting.

Still, we believe that DAOs can provide an excellent practical battleground for experimenting with different social choice and voting techniques in the real world.

\paragraph{Vote-buying / coercion:} There is a considerable literature on the notion of \textit{coercion-resistance} in end-to-end verifiable voting~\cite{juels2005coercion,delaune2006coercion,lueks2020voteagain}. Broadly speaking, coercion-resistance means that a voter cannot convince a would-be briber or coercer of how she voted. Influential proposed coercion-resistant voting systems include notably Civitas~\cite{clarkson2008civitas} and, more recently, MACI~\cite{Buterin-maci}. None of these definitions or system designs contemplate the risk of key encumbrance. Dark DAOs effectively break all of them.

\paragraph{Unauthorized credential delegation:} Daian et al.~put forth the notion of a Dark DAO---a DAO that aims to subvert voting in other DAOs---in~\cite{darkdaohack}. In related work, Matetic et al.~propose use of TEEs as a tool for secure credential delegation---which may be unauthorized~\cite{matetic2018delegatee}, and Puddu et al.~explore malicious uses, including  subversion of e-voting~\cite{puddu2019teevil}.

}
\section{Conclusion and Open Research Questions}
\label{sec:open_questions}

We have proposed \indexlong (\indexshort) as a new metric for DAO decentralization. \indexshort measures the entropy of voting blocs. It is in fact a framework into which it is possible to plug any desired method of clustering to identify blocs and any notion of entropy.

Evaluating \indexshort---instantiated with \clusterlong and min-entropy---we have proven a number of results that held shed light how a number of practices may impact DAO decentralization. We have also shown both in theory and through implementation of a practical system how Dark DAOs pose a potential long-term threat. 

Our work gives rise to a number of open research questions. A few deserve particular mention:

\begin{itemize}
    \item \textbf{Privacy:} Our results suggest the potential decentralizing effects of ballot secrecy, i.e., private voting. Existing verifiable end-to-end voting systems implement a one-vote-per-person policy~\cite{ali2016overview}. One open research question is whether token-weighted variants are possible. Additionally, we emphasize in our work the importance of collecting voting data to facilitate \indexshort estimates. How to harmonize these opposing goals represents a second research challenge.  
    \item \textbf{Forking and escape hatches:} DAOs may suffer catastrophic failures, as was famously the case with The DAO~\cite{TheDAO:2023}. Proposed remedies including forking / splitting, in which a new, quasi-independent or independent DAO is created and escape hatches, which are committee-controlled shutdowns~\cite{decentralized-escape-hatch}. How their existence and use impact decentralization are unclear and deserves study.
    \item \textbf{\indexshort impact:} \indexshort is designed to formalize a view of decentralization in DAOs reflected in the literature and in the views of practitioners. A natural question is what impact high \indexshort has on decision making in DAOs. Does it correlate with community growth, participation, and financial outcomes in DAOs? How does it relate to notions of democratic participation in non-blockchain settings? 
\end{itemize}


\section*{Acknowledgements}

This work was funded by NSF grant CNS-2112751 and by generous support from IC3 industry partners. Andr\'{e}s F\'{a}brega is funded in part by a Uniswap Foundation TLDR Fellowship.

Thanks to Oasis Labs for answering technical questions, Phil Daian for extensive discussions about Dark DAOs, and Sylvain Bellemare for suggesting the pronunciation ``vibe'' for $\indexshort$. 

\bibliographystyle{plain}
\bibliography{biblio}

\end{document}